\documentclass[11pt]{amsart}

\usepackage{amsmath,amssymb,amsthm,mathrsfs,amsbsy}
\usepackage{comment}
\usepackage[dvipsnames]{xcolor}
\usepackage{algorithm}
\usepackage{algpseudocode}
\definecolor{jade}{rgb}{0.0, 0.66, 0.42}
\usepackage[unicode,breaklinks=true,colorlinks=true, citecolor=jade, urlcolor=blue]{hyperref}
\newcommand{\WW}{\textcolor{black}}

\usepackage{graphicx}
\newcommand{\innerprod}[2]{\left\langle #1,#2\right\rangle}

\DeclareMathOperator*{\argmin}{arg min}

\usepackage[top=1in, bottom=1in, left=1in, right=1in, marginparwidth=1in, marginparsep=0.1in]{geometry}

\numberwithin{equation}{section}
\newtheorem{theorem}{Theorem}[section]
\newtheorem{proposition}[theorem]{Proposition}
\newtheorem{lemma}[theorem]{Lemma}

\theoremstyle{remark}

\definecolor{darkblue}{rgb}{0,0,0.7}

\newcommand{\Rnot}{\mathscr R_0}

\newcommand{\abs}[1]{\left| #1 \right|}

\newcommand{\pow}[2]{#1^{(#2)}}

\newcommand{\eps}{\varepsilon}

\newcommand{\omu}{(\mu^* - \mu)}

\newcommand{\R}{{\mathbb R }}

\newcommand{\donothing}[1]{{}}

\newcommand*{\defeq}{\mathrel{\vcenter{\baselineskip0.5ex \lineskiplimit0pt
                     \hbox{\scriptsize.}\hbox{\scriptsize.}}}%
     		     =}

\makeatletter
\newcommand{\xRightarrow}[2][]{\ext@arrow 0359\Rightarrowfill@{#1}{#2}}
\makeatother

\let\OLDthebibliography\thebibliography
\renewcommand\thebibliography[1]{
  \OLDthebibliography{#1}
  \setlength{\parskip}{1pt}
  \setlength{\itemsep}{1pt plus 0.3ex}
}

\renewcommand{\H}{\mathcal H}

\title{Optimal Control of an SIR Model with Noncompliance as a Social Contagion}
\author{Chloe Ngo}
\author{Christian Parkinson}
\author{Weinan Wang}
\date{\today}

\begin{document}

\begin{abstract}
We propose and study a compartmental model for epidemiology with human behavioral effects. Specifically, our model incorporates governmental prevention measures aimed at lowering the disease infection rate, but we split the population into those who comply with the measures and those who do not comply and therefore do not receive the reduction in infectivity. We then allow the attitude of noncompliance to spread as a social contagion parallel to the disease. We derive the reproductive ratio for our model and provide stability analysis for the disease-free equilibria. We then propose an optimal control scenario wherein a policy-maker with access to control variables representing disease prevention mandates, treatment efforts, and educational campaigns aimed at encouraging compliance minimizes a cost functional incorporating several cost concerns. Via careful analysis of the control-to-state map, we are able to prove existence of optimal controls. Our proof applies to dynamics which can be nonlinear in the control variables and general cost functionals including the case of $L^1$ control costs. We numerically resolve optimal strategies using the sequential quadratic Hamiltonian method, a relatively new numerical method for optimal control which is easy to implement and has good convergence theory, as we demonstrate. We test our model in several parameter regimes with specific interest in observing how the policy-maker's optimal strategies depend on their particular preferences which are expressed via design of different cost functionals.


\end{abstract}

\maketitle

\section{Introduction}

In this paper, we are interested in incorporating human behavioral effects into a Susceptible-Infected-Recovered (SIR) type compartmental model for epidemic spread of a generic infectious disease. At the outset of an epidemic, governments will enact public health protocols meant to slow disease spread. However, adherence with non-pharmaceutical intervention (NPI) measures such as mask-wearing, shelter-at-home mandates, and social distancing is certainly imperfect. For example, Dolan provides an interesting account of opposition to prolonged public health protocols during the 1918 Flu epidemic and the actions of the so-called ``anti-mask league" \cite{Unmasking}, and in a study by the Pew Research Center from June 4-10, 2020 during the height of the COVID 19 pandemic, 16\% of $9654$ Americans surveyed reported that they ``hardly ever" or ``never" wear a mask \cite{Pew}. Likewise, the Angus Reid Institute reported on July 15, 2020 that 45\% of 1503 Canadians surveyed responded that they ``rarely" or ``never" wear masks, and among the 1202 participants who reported that they do not always wear masks, 15\% reported that their main reason for not wearing a mask is that they believe masks are ineffective and 11\% reported their main reason for not wearing a mask as ``no one else is wearing them" \cite{Angus}. Other studies reported similar levels of nonadherence with several types of NPIs \cite{NC1,NC3,NC2}. From reports like this, it appears that there will always be nontrivial subpopulations who do not comply with NPI measures, and given evidence that NPI measures have a positive effect on epidemic outcomes \cite{NPIEffectiveness}, accounting for this noncompliance is key to increasing fidelity in mathematical models of epidemiology. 

Accordingly, we propose a compartmental model for epidemiology wherein the population is split into those who comply with public health protocols (dubbed the \emph{compliant} population), and those who do not comply (dubbed the \emph{noncompliant} population). The disease is assumed to spread slower among the compliant population. Social contagion theory \cite{SCTOverview} posits that attitudes, opinions and behaviors can spread among a population in a manner similar to that of an infectious disease, and has been used to model prescription and illicit drug use \cite{SCdrugs0, SCdrugs1, SCdrugs2}, alcohol consumption \cite{SCalcohol}, governmental corruption \cite{SCcorruption}, adverse mental health conditions such as depression \cite{SCdepression}, and public sentiment \cite{SC_publicSentiment} in a variety of contexts including opinions regarding anti-smoking legislation \cite{SCantismoking} and, importantly, vaccine and treatment uptake during epidemics \cite{SCvaccines0,SCvaccines}. Given this, and the aforementioned study wherein people reported not wearing masks because ``no one else is wearing them" \cite{Angus}, we suppose that the attitude of noncompliance with NPIs spreads throughout the population as a social contagion. \WW{This provides a framework for studying epidemic control with dynamic human behavioral consideration, as opposed to treating behaviors as fixed or externally-prescribed.} We further imagine a policy-maker with several different options for epidemic control who aims to minimize a cost functional which trades off costs for increased infections and increased noncompliance with economic costs for increased use of controls. 

Incorporation and assessment of the effects of heterogeneous behavior  into epidemiological models is becoming increasingly popular (see, for example, \cite{HB0,HB1,HB2,HB3,HB4}, or \cite{HB5} where the authors propose a rather general framework similar to the modeling we use below, though with linear transfer between behavior classes). Likewise, there are several recent studies of optimal control for epidemiological models in different modeling scenarios, such as optimal treatment or vaccination \cite{amalraj2025backward, MR3181992, engida2023optimal,MR3012899, zhou2014optimal}, optimal quarantine under a variety of limitations \cite{OC_limitQuarantine,dagher}, and assessment of the effects of contact heterogeneity \cite{OC_contactHeterogeneity} or the use of different types of cost functionals \cite{OC_differentCostFunctionals}. \WW{However, very few studies include both optimal control for epidemiology and mechanistic modeling of human behavior.} One exception is \cite{BPW} which takes a more theoretical look at a system of controlled partial differential equations describing an epidemic. Other studies incorporate control and human behavior more implicitly, by, for example, allowing infection rates to depend on ambient dynamics under the assumption that people will change their behavior based on local disease effects \cite{OC_behavioralResponse} or media coverage \cite{OC_withMediaCoverage}. 

\subsection{Our Contribution} \WW{The sociological contribution of our work is threefold. First, we consider optimal control of a compartmental epidemiological model which incorporates human behavior in an explicit, mechanistic manner. From a modeling standpoint, this allows adherence with NPIs to be viewed as a dynamic epidemiological quantity as opposed to, for example, a fixed parameter. Second, we study the simultaneous control of disease and behavioral transmission, and specifically investigate the manner in which resources should be allocated toward stunting disease spread and/or efforts to increase compliance with NPIs given different policy-maker preferences. Third, we demonstrate that different policy-maker preferences lead to vastly different optimal control strategies. In certain cases, slowing disease spread via synergistic use of traditional public health measures and efforts to increase compliance with NPIs is optimal, whereas in cases when the policy-maker is too cost-conscious or noncompliance spreads too quickly, controls will simply be ineffective.}

On the mathematical side, we prove existence of optimal controls in a generic scenario wherein (i) controls appear in the model nonlinearly, (ii) control costs are not required to be smooth or strongly convex (e.g. $L^1$ control costs are allowable), and (iii) the admissible control set contains discontinuous control maps. With small tweaks, our existence proof will go through for many biological models where (i)-(iii) are physically preferable, but are commonly relaxed (e.g. by the restriction to control affine systems and use of $L^2$ control costs) for ease of analysis. We resolve optimal controls via the sequential quadratic Hamiltonian (SQH) method \cite{SQH}. This is relatively new algorithm for optimal control which is black-box free and has a reasonable convergence theory, as demonstrated below. This distinguishes SQH from open source large scale optimization packages (used, for example, in \cite{MR3181992,dagher}) which are explainable to experts, but often operate as black-boxes for the end user, and also from the popular forward-backward sweep method (used, for example, in \cite{amalraj2025backward,OC_withMediaCoverage,engida2023optimal}) which is easy to implement but does not have a robust convergence theory. Use of the SQH method in the epidemiology literature is scant. It was first employed in \cite{SQHepidemiology}, where the authors demonstrate that it vastly outperforms the basic forward-backward sweep method, and it also appears in the preprint \cite{ParkRoy} in the context of control for a partial differential equation model for epidemics. Thus, the mathematical contribution of this manuscript goes further than the particular model we analyze, but also includes a general existence framework for optimal controls that applies to biologically realistic models, and the proliferation of a numerical method for optimal control which may be preferable to those commonly used in the epidemiology literature.

The remainder of the manuscript is organized as follows. In section \ref{sec:OurModel}, we introduce our model, derive the reproductive ratio, and discuss stability of disease-free equilibria. In section \ref{sec:OptControl}, we introduce the optimal control problem that we will solve and prove existence of optimal controls for a reasonably generic cost functional. In particular, our existence proof allows for $L^1$ control cost and discontinuous controls, which are physically meaningful for our model. In section \ref{sec:numericalMethods}, we describe the application of the SQH method \cite{SQH} to our problem, and prove convergence of the algorithm. In section \ref{sec:results}, we discuss the results of simulations of our model in different parameter regimes. Specifically, we make some qualitative observations regarding the optimal control strategies that result from different policy-maker preferences. Finally, section \ref{sec:conclusion} includes a brief conclusion and discussion of avenues for future work.

\section{An SIR Model with Behavioral Effects: Disease-Free Equilibria and Reproductive Ratios} \label{sec:OurModel}

As with all ordinary and partial differential equation (O/PDE) based epidemiological models, we build off of the classical SIR model of Kermack and McKendrick \cite{OG}. We assume a constant birth rate $b >0$ into the susceptible population, and that deaths occur proportional to population sizes with natural rate $\delta > 0$. The infection spreads via mass action with rate $\beta >0$ and infectious individuals recover with rate $\gamma > 0$. This leads to the SIR model 
 \begin{equation}\label{eq:basicSIR}
\begin{split}
	\frac{dS}{dt} &= b-\beta SI - \delta S,\\
	\frac{dI}{dt} &= \beta SI - (\gamma +\delta) I, \\
    \frac{dR}{dt} &= \gamma I - \delta R. 
	\end{split}
 \end{equation} As mentioned above, we incorporate human behavior in a manner inspired by social contagion theory. Each population is divided into a compliant portion (denoted by $S,I,R$) and a noncompliant portion (denoted by $S^*,I^*,R^*$), and we assume that, parallel to the disease, noncompliance spreads via a second mass action term. With this assumption, we propose the following model:
 \begin{equation} \label{eq:SIRwithCompliance} \begin{split}
	\frac{dS}{dt} &= (1-\xi) b - \beta (1-\alpha(t))S((1-\alpha(t))I + I^*) -(\mu^*-\mu(t)) SN^* + \nu(t) S^* - \delta S,\\
	\frac{dI}{dt} &= \beta (1-\alpha(t))S((1-\alpha(t))I + I^*)- (\gamma+\eta(t)) I - (\mu^*-\mu(t)) IN^* + \nu(t) I^* - \delta I, \\
	\frac{dR}{dt} &= (\gamma+\eta(t)) I - (\mu^*-\mu(t)) RN^* +\nu(t) R^* - \delta R,\\
	\frac{dS^*}{dt}&= \xi b -\beta S^*((1-\alpha(t))I+I^*) + (\mu^*-\mu(t)) SN^* - \nu(t) S^* -\delta S^* ,\\
	\frac{dI^*}{dt} &= \beta S^*((1-\alpha(t))I+I^*) - \gamma I^* +(\mu^*-\mu(t)) IN^* -\nu(t) I^* - \delta I^*, \\
	\frac{dR^*}{dt} &= \gamma I^* + (\mu^*-\mu(t)) RN^* - \nu(t) R^* - \delta R^* .
	\end{split}
	\end{equation} 
Here the birth rate $b$, death rate $\delta$, infectivity rate $\beta$ and recovery rate $\gamma$ have the same interpretation as in \eqref{eq:basicSIR}. New parameters are as follows: $\xi \in [0,1]$ is the probability with which newly introduced members of the population are noncompliant, and hence the birth rate $b$ is split between the $S$ and $S^*$ equations with contributions $(1-\xi)b$ and $\xi b$ respectively. The total noncompliant population is defined by $N^* = S^*+I^*+R^*$ and is involved in new mass action terms which cause transfer from $S,I,R$ to $S^*,I^*,R^*$ with baseline infectivity rate $\mu^*$ (the assumption being that anytime a compliant individual comes in contact with a noncompliant individual, they have some chance of becoming noncompliant). The remaining parameters $\alpha, \eta,\mu,\nu$ represent the policy-maker's attempts to stop the spread of the disease and/or noncompliance. While we will typically suppress their $t$-dependence, we have written them as functions of $t$ in \eqref{eq:SIRwithCompliance} to emphasize that these will be treated as control variables that the policy-maker can adjust throughout the course of the epidemic. Their respective interpretations and admissible values are: \begin{itemize}
    \item[(1)] $\alpha(\cdot) \in [0,\overline \alpha]$ - reduction in infectivity or mixing for the compliant population. Here $\overline \alpha \in [0,1]$ is the maximum decrease in mixing / infectivity that can be achieved through the use of NPIs. Note, $\alpha(t)$ appears in the mass action terms representing spread of infections; in these terms, the compliant subpopulations $S,I$ are multiplied by $(1-\alpha(t))$, while the noncompliant subpopulations $S^*,I^*$ are not. 
    \item[(2)] $\eta(\cdot) \in [0,\overline \eta]$ - increased recovery rate among compliant infectious population due to treatment seeking. Here $\overline \eta$ is the maximum achievable increase in recovery rate for the disease due to treatment.
    \item[(3)] $\mu(\cdot) \in [0,\overline \mu]$ - reduction in spread of noncompliance due to public health information initiatives. Here $\overline \mu \in [0,\mu^*]$ is the maximum achievable decrease in the spread rate for noncompliance.
    \item[(4)] $\nu(\cdot) \in [0,\overline \nu]$ - rate of recovery from noncompliance due to educational campaigns aimed at increasing compliance with NPIs. Here $\overline \nu$ is the maximum achievable recovery rate from noncompliance. 
    \end{itemize} 

    \WW{Using these controls, our model distinguishes between control efforts aimed at slowing disease spread and those aimed at increasing adherence with NPIs.}
    
    We define $U_{\text{ad}} = [0,\overline \alpha] \times [0,\overline \eta] \times [0,\mu^*] \times [0,\overline \nu]$ to be our set of admissible control actions. The uncontrolled scenario would be $\alpha(\cdot), \eta(\cdot),\mu(\cdot),\nu(\cdot) \equiv 0$, and the maximally controlled scenario would have each control variable identically equal to its maximum value. Note that we do not consider vaccination, which can be modeled as either an ordinary or ``singular" control variable as discussed in \cite{MR3012899}, opting instead to focus on the early stages of infection when vaccination is unavailable and NPIs play a larger role. In this way, our model is similar to the ones analyzed in \cite{BPB} in the ODE setting, \cite{PW3} in the stochastic ODE setting, and \cite{BPW,PW} in the PDE setting. However, \cite{BPB,PW,PW3} do not consider the optimal control problem, and the work in \cite{BPW} does not incorporate treatment as a control, and is also more focused on analytic results regarding existence and characterization of solutions to a PDE-constrained optimization problem. Here we are interested in empirically observing features of the optimal control maps in different parameter regimes. 

In the ensuing section, we will introduce a cost functional for the policy-maker, and evaluate optimal control strategies. Before that, we cover some standard considerations regarding compartmental epidemiological models; namely, nonnegativity of solutions, derivation of the reproductive ratio, and stability of disease-free equilibria. We note that this work is already contained in \cite{PW3}, except in that manuscript $\eta = 0$. Because $\eta > 0$ causes an asymmetry in the way $I$ and $I^*$ are treated, the derivations here are a bit messier, though no more difficult. Thus to avoid duplication, we will cover them tersely and refer the reader to \cite{PW3} for proofs and expanded reasoning.

\begin{proposition} \label{prop:exist}
    Assuming nonnegative initial conditions and $(\alpha,\eta,\mu,\nu) \in L^\infty([0,T]; U_{\text{ad}})$, \eqref{eq:SIRwithCompliance} admits a unique, nonnegative, a.e. solution on the time interval $[0,T].$
\end{proposition}

The existence and uniqueness follow directly from Caratheodory's theorem \cite[Ch. 1.6]{Roub}. The preservation of nonnegativity in proposition~\ref{prop:exist} follows since $\frac{dX}{dt}|_{X=0} \ge 0$ whenever all components are nonnegative, for each $X \in \{S,I,R,S^*,I^*,R^*\}.$ We also note that there is a total population bound. Indeed, defining $N_{\text{total}} = S+I+R+S^*+I^*+R^*$, we see \begin{equation} \label{eq:TOTALPOP}N_{\text{total}}' = b-\delta N_{\text{total}} \,\,\,\,\, \implies \,\,\,\,\, N_{\text{total}}(t) = \frac{b}{\delta} + \bigg(N_{\text{total}}(0)-\frac b \delta\bigg)e^{-\delta t} \le \max\left\{N_{\text{total}}(0),\frac b \delta\right\}\end{equation} We can assume without loss of generality that $N_{\text{total}}(0) = 1$, whereupon alternately setting $\tfrac b \delta \ge 1$ or $\tfrac b\delta < 1$ would correspond to situations where the population is increasing or where the population is strictly decreasing. 

For the remainder of this section, we treat $\alpha,\eta,\mu,\nu$ as \emph{fixed, nonnegative, constant} values for the sake of finding disease free equilibrium (DFE) points. The DFEs for our system are equilibrium points $x^* = (s,i,r,s^*,i^*,r^*)$ with $i=i^*=r=r^* = 0$ and $s,s^* \ge 0$. There are two DFEs if $\xi = 0$ (meaning that all newly introduced members of the population are compliant). These are given by \begin{equation}
    \label{eq:DFEs12} 
    \begin{split}
        s_1 &= \tfrac b \delta, \\ s^*_1 &= 0,
    \end{split} \hspace{2cm} \text{ and } \hspace{2cm} \begin{split}
        s_2 &= \tfrac{\nu + \delta}{\mu^* - \mu}, \\ s^*_2 &= \tfrac b \delta - \tfrac{\nu + \delta}{\mu^* - \mu},
    \end{split} 
\end{equation} The first is always physically meaningful and represents the ``best case scenario" for the policy-maker wherein the whole population is compliant. The second is physically meaningful and distinct only if $\tfrac b \delta > \frac{\nu+\delta}{\mu^* - \mu}.$ Since $\mu \in [0,\mu^*]$, the denominator of $\frac{\nu+\delta}{\mu^* - \mu}$ is nonnegative, so in the case that $\mu \to \mu^*$, the second DFE is no longer meaningful.  

If instead $\xi \in (0,1]$ so that some portion of the newly introduced population is noncompliant, there is a unique DFE \begin{equation}  \label{eq:DFEs3}\begin{split}
    s_3 &= \tfrac 1 2\left( \tfrac{b}{\delta} + \tfrac{\delta + \nu}{\mu^* - \mu} - \sqrt{\left(\tfrac{b}{\delta} - \tfrac{\delta + \nu}{\mu^* - \mu} \right)^2 + \tfrac{4\xi b}{\mu^* - \mu} } \right), \\
    s^*_3&= \tfrac 1 2\left( \tfrac{b}{\delta} - \tfrac{\delta + \nu}{\mu^* - \mu} + \sqrt{\left(\tfrac{b}{\delta} - \tfrac{\delta + \nu}{\mu^* - \mu} \right)^2 + \tfrac{4\xi b}{\mu^* - \mu} } \right). \end{split}
\end{equation} In the limit as $\xi \to 0$, we find $(s_3,s_3^*) \to (s_1,s_1^*)$ when $\tfrac b \delta \le \frac{\nu+\delta}{\mu^* - \mu}$, and $(s_3,s_3^*) \to (s_2,s_2^*)$ when $\tfrac b \delta > \frac{\nu+\delta}{\mu^* - \mu}$. A special case of \eqref{eq:DFEs3} is the ``worst case scenario" for the policy-maker wherein $\xi = 1$  and $\nu = 0$, meaning that all newly introduced members of the population are noncompliant and noncompliance is a permanent state. In this case, the DFE is $s_3 = 0, s^*_3 = \tfrac b \delta$.

The \emph{reproductive ratio} $\mathscr R_0(s,s^*)$ corresponding to a DFE $x_{s,s^*} = (s,0,0,s^*,0,0)$ can be computed by the next generation matrix method \cite{NextGenMat1,NextGenMat2,NextGenMat3}: \begin{equation} \label{eq:Rnot} 
    \mathscr R_0(s,s^*) = \frac{\beta\Big[(1-\alpha)^2\Big(1+ \tfrac{\alpha\mu s^*/(1-\alpha)-\eta}{\gamma + \eta + \nu + \delta + (\mu^*-\mu)s^*}\Big)s + \Big(1-\tfrac{\alpha \nu}{\gamma + \eta + \nu + \delta + \omu s^*}\Big)s^*\Big] }{\gamma + \delta + \frac{\eta \nu}{\gamma + \eta + \nu + \delta + (\mu^*  -\mu)s^*}}.
\end{equation} This has a reasonably natural interpretation, with the main complicating factor being the difference in recovery rates between the compliant and noncompliant populations due to the compliant population seeking treatment represented by $\eta$. In the worst case scenario, where $\xi = 1, \nu = 0$ so that the DFE is $(s,s^*) = (0,\tfrac b \delta)$, we find $\mathscr R_0(0,\tfrac b \delta) = \frac{b}{\delta}\cdot \frac{\beta}{\gamma + \delta}.$ This is the reproductive ratio for the basic SIR model \eqref{eq:basicSIR}, which signifies that in this worst case scenario, the policy-maker's control efforts do nothing to stop the spread of the disease. In the best case scenario, $\xi = 0$ and $\tfrac b \delta \le \frac{\nu+\delta}{\mu^* - \mu}$ so that the only DFE is $(s,s^*) = (\tfrac b \delta, 0)$ and we find that $\mathscr R_0(\tfrac b \delta,0) = \tfrac b \delta \cdot\frac{\beta(1-\alpha)^2}{\gamma + \eta + \delta }.$ This decreases the reproductive ratio when compared with the basic SIR model in two ways: $(1-\alpha)^2$ in the numerator results from the reduction in infectivity among the compliant populations and the $\eta$ in the denominator results from the increased recovery rate among the compliant infected population.  In general, $\mathscr R_0(s,s^*)$ can then be viewed as a very complicated interpolation between the two points wherein the compliant susceptible population $s$ is receiving some reduction in infectivity signified by the $(1-\alpha)^2$, whereas the noncompliant susceptible population $s^*$ receives no such benefit. 

We make some natural observations regarding $\mathscr R_0(s,s^*)$. 

\begin{proposition} \label{prop:sanity}  Define $\mathscr R_0(s,s^*)$ as in \eqref{eq:Rnot} for nonnegative $s,s^*$.
\begin{itemize} 
\item[(i)] Holding other values constant, $\mathscr R_0(s,s^*)$ is decreasing in each of the control variables $\alpha, \eta, \mu,\nu$. 
    \item[(ii)] The function $s^* \mapsto \mathscr R_0(\tfrac b \delta - s^*, s^*)$ is increasing for $s^* \in [0,\tfrac b \delta]$. 
    \end{itemize}
\end{proposition}

\noindent{\bf Remark.} This proposition is a sanity check to verify that our model behaves as expected. The two points can be interpreted as follows: (i) more control leads to a lower rate of disease spread, and (ii) when we are at the equilibrium population level of $\tfrac b \delta$, more noncompliance leads to a higher rate of disease spread. We omit the proof because it is elementary, if tedious. It follows most readily from the observation that the function $z\mapsto \frac{c_1+c_2z}{c_3+c_4z}$ is increasing for $z\ge 0$ when $c_2c_3 \ge c_1c_4$, and decreasing for $z\ge 0$ when $c_2c_3\le c_1c_4$. Clearing the nested denominators, one can write $\mathscr R_0(s,s^*)$ in this form with $\alpha,\eta,\mu,\nu$ alternately playing the role of $z$, and verify the correct inequality. \\

To conclude this section, we state stability results for our DFE.

\begin{theorem} \label{thm:DFEStability}
    Considering the DFE values given in \eqref{eq:DFEs12}, \eqref{eq:DFEs3} and $\Rnot(s,s^*)$ as defined in \eqref{eq:Rnot}, we have the following. \begin{itemize}
        \item[(i)] If $\xi = 0$ and $\tfrac b \delta\cdot \tfrac{\mu^*-\mu}{\nu + \delta} < 1$, then $(s_1,s_1^*) = (\tfrac b \delta,0)$ is asymptotically stable when $\Rnot(\tfrac b \delta, 0) < 1$ and unstable when  $\Rnot(\tfrac b \delta, 0) > 1$. 
        \item[(ii)] If $\xi = 0$ and $\tfrac b \delta\cdot \tfrac{\mu^*-\mu}{\nu + \delta} > 1$, then $(s_2,s_2^*) = (\tfrac{\nu+\delta}{\mu^* - \mu},\tfrac b \delta-\tfrac{\nu+\delta}{\mu^* - \mu})$ is asymptotically stable when $\Rnot(\tfrac{\nu+\delta}{\mu^* - \mu},\tfrac b \delta-\tfrac{\nu+\delta}{\mu^* - \mu}) < 1$ and unstable when $\Rnot(\tfrac{\nu+\delta}{\mu^* - \mu},\tfrac b \delta-\tfrac{\nu+\delta}{\mu^* - \mu}) > 1$
        \item[(iii)] If $\xi \in (0,1]$, then $(s_3,s_3^*)$ as defined in \eqref{eq:DFEs3} is asymptotically stable when $\Rnot(s_3,s_3^*) < 1$ and unstable when $\Rnot(s_3,s_3^*)>1$.
    \end{itemize}
\end{theorem}

\noindent{\bf Remark.} In each case, that the condition on the reproductive ratio determines local asymptotic stability or instability follows directly from \cite[Theorem 2]{NextGenMat3} once the assumptions of that theorem are verified. However, the stability is actually global, in the sense that under the conditions of each statement, \emph{any} solution of \eqref{eq:SIRwithCompliance} will tend toward the given equilibrium point regardless of whether the initial condition is nearby the equilibrium point. This global version can be proven using appropriately constructed Lyapunov functions. We omit the proof because it is quite long and technical, and is contained in \cite{PW3}. \\

\WW{To conclude this section, we remind the reader of the biological interpretation of these results. We note that one can look solely at the compliance-noncompliance transmission by considering $N=S+I+R, N^*=S^*+I^*+R^*$ and when $\xi = 0$ as in theorem \ref{thm:DFEStability}(i),(ii) these quantities obey simple susceptible-infected-susceptible (SIS) type dynamics. The quantity $\tfrac b \delta\cdot \tfrac{\mu^*-\mu}{\nu + \delta}$ is then the reproductive ratio for the behavioral dynamics, and the equilibrium points $(s_1,s^*_1)$ and $(s_2,s^*_2)$ are the \emph{noncompliance free} and \emph{noncompliance endemic} equilibria respectively. It may also be possible to identify \emph{disease endemic} equilibria. Preliminary exploration in that direction seems to indicate that disease endemic equilibria correspond to positive roots of a quartic polynomial whose coefficients are rather lengthy expressions involving all the ambient parameters (though further simplification is likely possible). However, for a policy-maker who deems disease suppression important, guiding the system to a disease endemic equilibrium is not likely to be a desirable outcome, so from the perspective of fixed-horizon time optimal control it is less important to identify these. By contrast, guiding the system to a DFE is likely to be a desirable outcome, and theorem \ref{thm:DFEStability} could assist a policy-maker in doing so. }

\WW{
The preceding disease-free equilibrium (DFE) and reproductive-ratio ($\mathcal R_0$) calculation should be interpreted as an autonomous, constant-control threshold analysis. In the numerical section below the controls vary in
time, but at each time they move the instantaneous system toward a parameter regime in which
the relevant disease-free equilibrium is more or less stable. Thus, the purpose of the controls is not
only to reduce the instantaneous force of infection, but also to change the behavioral composition
of the population so that the disease-free regime becomes easier to maintain.
}

\WW{Proposition~\ref{prop:sanity} provides a first sensitivity analysis of this threshold quantity. Holding the disease-free composition $(s,s^*)$ fixed, the reproductive ratio decreases with each of the four controls $\alpha,\eta,\mu,\nu$. The interpretation is direct: NPIs reduce transmission among compliant individuals, treatment increases removal from the compliant infected class, and the two behavior-directed controls reduce the ability of noncompliance to create a reservoir in which disease-directed controls are less effective. Conversely, along the equilibrium population constraint $s+s^*=b/\delta$, the reproductive ratio is increasing in $s^*$. Hence the behavioral state of the susceptible population is itself an epidemiological risk factor: a larger noncompliant susceptible population increases the number of secondary infections generated by a typical infectious individual.}

\WW{This also shows how the influx parameter $\xi$ enters the threshold dynamics. For $\xi\in(0,1]$ and $\mu<\mu^*$, the noncompliant component of the disease-free equilibrium satisfies
\[
\frac{\partial s_3^*}{\partial \xi}
=\frac{b}{(\mu^*-\mu)\sqrt{\left(\frac{b}{\delta}-\frac{\delta+\nu}{\mu^*-\mu}\right)^2+\frac{4\xi b}{\mu^*-\mu}}}>0.
\]
Together with Proposition~\ref{prop:sanity}(ii), this implies that increasing the baseline fraction of newly introduced noncompliant individuals increases the reproductive ratio through its effect on the disease-free population composition. This is the analytic threshold counterpart of the numerical observation in Scenario~3 of section~\ref{sec:results}: once the baseline tendency toward noncompliance is sufficiently large, behavior-directed controls lose much of their leverage.}

\section{Optimal Control of our Model} \label{sec:OptControl}

In this section, we formulate an optimal control problem for \eqref{eq:SIRwithCompliance}. For the purpose of simulations in section~\ref{sec:results}, we will design a cost functional incorporating three desires for the policy-maker: (1) to minimize the total number of infections throughout the course of the epidemic, (2) to minimize the total amount of noncompliance throughout the epidemic, (3) to minimize control costs. However, for the theory in this section, we consider a reasonably general control formulation. \WW{Accordingly, this section is quite mathematically technical, as necessitated by the nonlinear incorporation of controls in the dynamics and the generic structure of the cost functional which allows for the case of $L^1$ control cost. These are integral facets of posing a biologically meaningful control problem for our model, but they require a much more careful analysis than, for example, control-affine dynamics and $L^2$ control cost.} 

For brevity, we write the control map as $u(\cdot) \defeq (\alpha(\cdot),\eta(\cdot),\mu(\cdot),\nu(\cdot))$ and the corresponding trajectory satisfying \eqref{eq:SIRwithCompliance} as $x \defeq (S,I,R,S^*,I^*,R^*).$ With this notation, we define a generic cost functional
\begin{equation} 
\label{eq:costFunc} J(x,u) = \int^T_0 (r(x(t)) + \ell(u(t))) dt
\end{equation} where $r:\R^6_{\ge 0} \to \R$ and $\ell:\R^4\to\R$. We state assumptions on $r,\ell$ in theorem~\ref{thm:optControlsExist}, but we make special note that we do not require strict or strong convexity. \WW{The exact choices of $r,\ell$ indicate the preferences of the policy-maker. For example, a more public-health oriented policy-maker would choose $r$ which assigns large values to states with higher $I$ and $I^*$, whereas an economically stingy policy-maker would choose $\ell$ to take large values for larger values of the controls.} From the above, our admissible control set is \begin{equation}\label{eq:controlSet}
U_{\text{ad}} = \{(\alpha,\eta,\mu,\nu) \in \R^4 \, : \, 0 \le \alpha \le \overline \alpha, \, 0 \le \eta \le \overline \eta, \, 0 \le \mu \le \overline \mu, \, 0 \le \nu \le \overline \nu\}
\end{equation} where $\overline \alpha,\overline \eta, \overline \mu, \overline \nu$ are fixed positive constants. We then consider control functions in \begin{equation} \label{eq:admissibleControlFuncs}
\mathscr U_{\text{ad}} = \{u(\cdot) \in BV([0,T]; U_{\text{ad}}) \, : \, TV(u_i) \le M, \,\, i = 1,\ldots,4\}
\end{equation} for some universal constant $M$. That is, each individual control $\alpha(\cdot),\eta(\cdot),\mu(\cdot),\nu(\cdot)$ is in $BV[0,T]$, the space of real-valued functions of bounded variation. To recall, this is the set of functions $z:[0,T]\to \R$ such that \begin{equation}\label{eq:TV}TV(z) \defeq \sup_{g \in C_c^1[0,T], \|g\|_\infty\le 1} \int^T_0 z(t)g'(t)dt < \infty.\end{equation}  We remark on this choice of control functions after the proof of theorem~\ref{thm:optControlsExist}. Our optimal control problem is then \begin{equation} \label{eq:optControlProb} 
\min_{u \in \mathscr U_{\text{ad}}} J(x,u) \,\,\,\, \text{ subject to } \,\, \frac{dx}{dt} = f(x,u), \,\,\,\, x(0) = x_0
\end{equation} where $f(x,u)$ is defined by the right-hand side of \eqref{eq:SIRwithCompliance} and $x_0 \in \R^6$ has nonnegative components whose sum is bounded by $\tfrac b \delta$. In this notation, we note that the unique a.e. solution $x$ of \eqref{eq:SIRwithCompliance} corresponding to a control map $u \in \mathscr U_{ad}$ can be characterized as the unique function $x \in (L^2[0,T])^6$ satisfying the initial condition and \begin{equation}
    \label{eq:weakForm}
    \int^T_0 \langle x(t),\phi'(t)\rangle dt = \int^T_0 \langle f(x(t),u(t)),\phi(t)\rangle dt, \,\,\,\,\, \forall \phi \in C^1_c([0,T];\R^6).
\end{equation}

Because our admissible control set $U_{\text{ad}}$ is compact, $\mathscr U_{\text{ad}} \subset (L^p[0,T])^4$ for any $p \in [1,\infty)$ or $p = \infty.$ For simplicity, we prove the following existence theorem using the direct method considering convergence in the $L^2$ topology, though with small tweaks this argument would go through with other choices.

\begin{theorem} \label{thm:optControlsExist}
    Assuming $r$ and $\ell$ are proper, convex, and lower-semicontinuous and $x_0$ has nonnegative components, there exists an optimal control $u^* \in \mathscr U_{\text{ad}}$ and corresponding trajectory $x^*$ which solve the minimization problem \eqref{eq:optControlProb}. 
\end{theorem}

\begin{proof}
    By proposition~\ref{prop:exist}, there is a well-defined control-to-state operator $\mathcal S : \mathscr U_{\text{ad}} \to (H^1[0,T])^6$ so that for $u \in \mathscr U_{\text{ad}}$, $\mathcal S(u) = x \in (H^1[0,T])^6$ is the corresponding solution of \eqref{eq:SIRwithCompliance}. Note that the basic existence theory will only guarantee that $x \in (L^2[0,T])^6$, but then by the population bound \eqref{eq:TOTALPOP}, the right hand side of each equation in \eqref{eq:SIRwithCompliance} defines a bounded, and hence $L^2[0,T]$, function so $x \in (H^1[0,T])^6$ by a bootstrapping argument. Furthermore $\mathcal S(u) = x$ remains uniformly bounded in $(H^1[0,T])^6$ among choices of $u \in \mathscr U_{\text{ad}}.$ We segment the proof into a few steps. In all following statements, we assume $\mathscr U_{\text{ad}}$ and $(H^1[0,T])^6$ are given the $L^2$ topology.\\

    {\bf Step 1: $\mathcal S$ is weakly closed.}\\

    Suppose $u^{(k)} \in \mathscr U_{\text{ad}}$ is a sequence such that $u^{(k)} \rightharpoonup u$ in $(L^2[0,T])^4$ and $\mathcal S(u^{(k)}) \rightharpoonup x$ in $(L^2[0,T])^6.$ Since $\mathscr U_{\text{ad}}$ is  convex, it suffices to prove that it is strongly closed in $L^2$, and this will imply that it is weakly closed \cite[Ch. 5.1]{Conway}. That $\mathscr U_{\text{ad}}$ is closed under strong $L^2$ limits follows because if $z^{(k)} \to z$ strongly in $L^2$, then H\"older's inequality gives $$\abs{\int^T_0z^{(k)}(t)g'(t)dt - \int^T_0 z(t)g'(t)dt} \le \|z^{(k)}-z\|_{L^2[0,T]}\|g'\|_{L^2[0,T]} \to 0$$ for any $g \in C_c^1[0,T]$, and this proves that $TV(z^{(k)}) \to TV(z)$ for total variation as defined in \eqref{eq:TV}. In particular, strong $L^2$ limits respect bounds on total variation. Hence $\mathscr U_{\text{ad}}$ is weakly closed and so $u \in \mathscr U_{\text{ad}}$.  Further, since $\{u^{(k)}\}$ is of uniformly bounded total variation, Helly's selection theorem \cite[Thm. 5.5]{EvansMeasure} allows us to pass to a subsequence along which $u^{(k)}\to u$ pointwise a.e. in $[0,T].$ (We do so without renaming the sequence, as we do whenever we mention a subsequence in the ensuing work.)  
    
    Now $\mathcal S(u^{(k)})$ is a bounded sequence in the Hilbert space $(H^1[0,T])^6$, so it has subsequence converging weakly to some $\overline x \in (H^1[0,T])^6$. By the compact embedding $H^1[0,T] \hookrightarrow L^2[0,T]$ \cite[Ch. 5.7]{Evans}, passing to a further subsequence, we can ensure strong convergence to $\overline x$ in the $(L^2[0,T])^6$ topology. By uniqueness of limits, we must have $\overline x = x$, and thus passing to an even further subsequence, we have $\mathcal S(u^{(k)}) \to x$ pointwise a.e.  
    
    Finally, for each $k$, using the weak formulation \eqref{eq:weakForm}, we have \begin{equation} \label{eq:gonnaTakeALimit}\int^T_0 \langle [\mathcal S(u^{(k)})](t),\phi'(t)\rangle dt = \int^T_0 \langle f([\mathcal S(u^{(k)})](t),u^{(k)}(t)),\phi(t)\rangle dt, \,\,\,\,\, \forall \phi \in C^1_c([0,T];\R^6). \end{equation} As $k\to \infty$, using strong $L^2$ convergence on the left hand side of \eqref{eq:gonnaTakeALimit}, and pointwise a.e. convergence plus the dominated convergence theorem on the right hand side of \eqref{eq:gonnaTakeALimit}, we see $$\int^T_0 \langle x(t),\phi'(t)\rangle dt = \int^T_0 \langle f(x(t),u(t)),\phi(t)\rangle dt, \,\,\,\,\, \forall \phi \in C^1_c([0,T];\R^6), $$ showing that $x$ is the solution of \eqref{eq:SIRwithCompliance} corresponding to $u$. That is $\mathcal S(u) = x$, and $\mathcal S$ is weakly closed, which completes Step 1.\\

    {\bf Step 2: $\mathcal S$ is weak-strong continuous.}\\

    Suppose $u^{(k)} \rightharpoonup u$ weakly in $L^2$. We need to prove that $\mathcal S(u^{(k)}) \to \mathcal S(u)$ in $L^2.$ Once again, $\{\mathcal S(u^{(k)}\}$ is a uniformly bounded sequence in $(H^1[0,T])^6$, and so any subsequence $\{\mathcal S(u^{(k_\ell)}\}$ is still uniformly bounded $(H^1[0,T])^6$. By the compact embedding $H^1[0,T]\hookrightarrow L^2[0,T]$, $\{\mathcal S(u^{(k_\ell)}\}$ has a further subsequence $\{S(u^{(k_{\ell_m})}\}$ converging strongly (hence also weakly) to some function $x \in (L^2[0,T])^6$. By weak closure, we must have $\mathcal S(u) = x$, so we have proven that every subsequence of $\{\mathcal S(u^{(k)}\}$ has a further subsequence converging strongly to $\mathcal S(u)$ in $(L^2[0,T])^6$. This implies that $\mathcal S(u^{(k)})\to \mathcal S(u)$ strongly in $(L^2[0,T])^6$ by Urysohn's subsequence principle \cite[Ch. 2.1.17]{TaoMeasure}. \\

    {\bf Step 3: The cost functional is lower-semicontinuous so \eqref{eq:optControlProb} has a solution.} \\

    Because the control-to-state operator is well-defined, \eqref{eq:optControlProb} is equivalent to the reduced problem \begin{equation}\label{eq:reduced} 
        \min_{u \in \mathscr U_{\text{ad}}} \mathcal J(u) \defeq J(\mathcal S(u),u).
    \end{equation} Since $r,\ell$ are proper, convex, and lower-semicontinuous, they are bounded below, and since $u, \mathcal S(u)$ are bounded functions, this implies $\mathcal J(u)$ is bounded from below and thus has a finite infimum $\mathcal J_{0}$. Let $\{u^{(k)}\}$ be a minimizing sequence of controls, so that $\mathcal J(u^{(k)}) \to \mathcal J_{0}.$ Since $\{u^{(k)}\}$ is a bounded sequence in the Hilbert space $(L^2[0,T])^4$, we can pass to a subsequence that converges weakly to some $u \in \mathscr U_{\text{ad}}$. By weak closure and weak-strong continuity, $\mathcal S(u^{(k)}) \to x \defeq \mathcal S(u)$ strongly in $(L^2[0,T])^6.$ Since $r,\ell$ are proper, convex, and lower-semicontinuous, the maps $$x \mapsto \int^T_0 r(x(t))dt, \,\,\,\,\,\,\,\, u \mapsto \int^T_0 \ell(u(t))dt $$ are weakly lower-semicontinuous, and thus \begin{align*} \mathcal J_0 \le \mathcal J(u) = J(x,u) &= \int^T_0 (r(x(t)) + \ell(u(t))) dt \\ &\le \liminf_{k\to\infty} \int^T_0 \Big((r([\mathcal S(u^{(k)})](t)) + \ell(u^{(k)}(t))\Big)dt = \liminf_{k\to\infty} \mathcal J(u^{(k)}) = \mathcal J_0. \end{align*} Thus $u \in U_{\text{ad}}$ and the associated solution $x$ of \eqref{eq:SIRwithCompliance} form a solution of \eqref{eq:optControlProb}.
\end{proof}

We prove one further regularity result for the control-to-state map. While this result was not necessary in the proof above, it will have bearing in section~\ref{sec:numericalMethods}. 

\begin{lemma} \label{lem:CSlip} The control-to-state operator $\mathcal S:\mathscr U_{\text{ad}} \to (H^1[0,T])^6$ is Lipschitz continuous when both the source and target spaces are given the $L^2$ topology. 
\end{lemma}

\begin{proof}
    Let $u,v \in \mathscr{U}_{\text{ad}}$ be two control maps, and let $x=\mathcal S(u)$ and $y = \mathcal S(v)$ be the associated solutions of \eqref{eq:SIRwithCompliance}. Then for any $t \in [0,T],$ $$\abs{x(t) - y(t)}_2 \le \int^t_0 \abs{f(x(s),u(s))-f(y(s),v(s))}_2 ds.$$ By proposition~\ref{prop:exist}, solutions of \eqref{eq:SIRwithCompliance} are uniformly bounded over choices of controls, and the controls themselves are uniformly bounded. Thus since $f(x,u)$ is smooth in both $x$ and $u$, it is, in particular, Lipschitz continuous, so $$\abs{x(t) - y(t)}_2 \le L\int^t_0 \Big(\abs{x(s)-y(s)}_2 + \abs{u(s)-v(s)}_2\Big)ds.$$ Thus Gr\"onwall's inequality gives $$\abs{x(t)-y(t)}_2 \le Le^T \int^t_0 \abs{u(s)-v(s)}_2ds \le Le^T \int^T_0 \abs{u(s)-v(s)}_2ds.$$ Squaring this and invoking Jensen's inequality yields $$\abs{x(t)-y(t)}^2_2 \le L^2Te^{2T} \|u-v\|^2_{(L^2[0,T])^4}$$ and integrating on $[0,T]$ shows that $$\|x-y\|^2_{(L^2[0,T])^6} \le (LTe^{T})^2 \|u-v\|^2_{(L^2[0,T])^4}$$ proving the claim.
\end{proof}

\noindent {\bf Remark.} A remark is in order regarding the specific choice of our control set $\mathscr U_{\text{ad}}$. If the dynamics are control affine, it is  enough to assume $u \in (L^2[0,T])^4$. However, our dynamics are not control affine. Thus in Step 1 of the proof, some higher regularity is required to ensure a stronger notion of convergence of $u^{(k)} \to u$, and allow us to pass to limit in \eqref{eq:gonnaTakeALimit}. The proof above would go through a bit more easily if we assume even more regularity--for example $u \in (H^1[0,T])^4$--but this is highly undesirable, because it would preclude the possibility of discontinuous control maps. Discontinuous controls arise naturally when enforcing $L^1$ control costs since these commonly lead to bang-bang control strategies (see \cite[Ch. 17]{Lenhart} or \cite{BangBang2,BangBang3}). \WW{For example, in the context of epidemiology, this corresponds to the realistic scenario of maximally enforcing NPIs or behavioral control strategies with fixed start and end dates, as opposed to slowly ramping up to maximal enforcement and then slowly abating the enforcement when control is no longer desired.} The assumption that the controls have uniformly bounded total variation allows for the use of Helly's selection theorem \cite[Thm. 5.5]{EvansMeasure} to pass to subsequence which converges pointwise a.e., while still maintaining strong closure of $\mathscr U_{\text{ad}}$ in the $L^2$ topology, and also allowing for discontinuous control maps. \\

Given existence of optimal controls, a necessary condition for optimality is provided by the Pontryagin minimum principle (PMP). Using state variables $x = (S,I,R,S^*,I^*,R^*)$, costate variables $p = (p_S,p_I,p_R,p_{S^*},p_{I^*},p_{R^*})$, and control variables $u = (\alpha,\eta,\mu,\nu),$ we define the Hamiltonian \begin{equation}
    \label{eq:Hamiltonian} \begin{split}
    \mathcal H(x,p,u) &= \langle p,f(x,u)\rangle + r(x) + \ell(u),
    \end{split}
\end{equation} where $r,\ell$ are as in \eqref{eq:costFunc} and $f(x,u)$ is defined by the right-hand side of \eqref{eq:SIRwithCompliance}. Then the PMP states that if $ x(\cdot), u(\cdot)$ are an optimal state and control pair for \eqref{eq:optControlProb}, then there is a corresponding costate trajectory $ p(\cdot)$ such that \begin{align}
\frac{d x}{dt} &= \hphantom{-}\nabla_p \H( x, p, u), \,\,\,\,\,\,\,\,\,\,  x(0) = x_0, \label{eq:stateDyn} \\
\frac{d p}{dt} &= - \nabla_x \H( x,  p,  u), \,\,\,\,\,\,\,\,\,\,  p(T) = 0, \label{eq:costateDyn}\\
 u(t) &= \argmin_{\hat u \in U_{\text{ad}}} \H( x(t), p(t),\hat u),  \,\,\,\,\, \text{ for all } \,\,\, t \in [0,T]. \label{eq:OptCondition}
\end{align} We direct the reader to several classical sources for proofs and discussions regarding the PMP; for example, \cite[\S5]{Athans}, \cite[\S4]{Liberzon}, \cite[\S1,2]{Pontryagin}, \cite[\S12]{CTTextbook}. We note that \eqref{eq:stateDyn} is exactly the original dynamics \eqref{eq:SIRwithCompliance}. We refer to \eqref{eq:costateDyn} as the \emph{costate dynamics}, and \eqref{eq:OptCondition} as the \emph{optimality condition}. The condition $p(T) = 0$ in \eqref{eq:costateDyn} is called the \emph{transversality condition} (discussed in the aforementioned sources), and takes this simplified form because we consider a fixed time, free end-point problem with no terminal cost. 

In the ensuing section, we describe numerically resolving optimal controls via the sequential quadratic Hamiltonian method \cite{SQH}. 

\section{Numerical Methods} \label{sec:numericalMethods}

In this section, we consider the numerical resolution of optimal control maps via the sequential quadratic Hamiltonian (SQH) method \cite{SQH}. Because this is a fairly new numerical method, and to make this manuscript as self-contained as possible, we discuss some basic facets of the method and prove convergence of the iteratively defined control maps. A full discussion of the method with worked examples can be found in \cite{SQH}. 

The SQH method is an iterative method for resolving control maps in a variety of scenarios involving optimal control of ordinary or partial differential equations. The basic strategy is to successively update the approximation of the control maps using quadratic perturbations to the true Hamiltonian \eqref{eq:Hamiltonian}. Specifically, for $\eps > 0$, define the auxiliary Hamiltonian \begin{equation} \label{eq:newControls}\mathcal H_{\eps}(x,p,u,\overline u) \defeq \mathcal H(x,p,\overline u) + \eps\abs{u-\overline u}^2_2.\end{equation} The quadratic term above introduces relaxation, so that solving $$u^{\text{new}}(t) = \argmin_{w\in U_{\text{ad}}} \mathcal H_\eps(x(t),p(t),u(t),w)$$ results in $u^{\text{new}}$ which balances minimizing the Hamiltonian with a penalty for straying too far from the currently employed control strategy $u$. 

Given this perturbed Hamiltonian, the full SQH algorithm is detailed in algorithm~\ref{alg:SQHmethod}. The basic strategy of the algorithm is to successively update $x, p$ and $u$ by solving the state and co-state equations and resolving controls via \eqref{eq:newControls}. One significant novelty is the inclusion of step (v) in the iteration, whereupon one checks the condition \begin{equation} \label{eq:check} J\left(x^{(k+1)},u^{(k+1)}\right)-J\left(x^{(k)},u^{(k)}\right)\le -\rho \, \tau.\end{equation} If this condition is met, it is determined that the new controls $u^{(k+1)}$ have achieved a sufficient decrease in the cost functional. In this case, the new controls $u^{(k+1)}$ are accepted, $\eps$ is decreased by setting $\eps \leftarrow \zeta \eps$ for some $\zeta \in (0,1)$, and the iteration continues with the updated auxiliary Hamiltonian. If the condition is not met, there was not sufficient decrease in the cost functional, so $\eps$ is increased by setting $\eps \leftarrow \lambda \eps$ for some $\lambda > 1$ and $u^{(k+1)}$ is recomputed before continuing. Here $\lambda,\zeta$ are tunable parameters, as is $\rho$ in \eqref{eq:check}, which determines the level to which $u^{(k+1)}$ must decrease the cost functional to be considered acceptable.  

%
%
\begin{algorithm}[t!]
\caption{Sequential Quadratic Hamiltonian Method for Resolving Optimal Control Maps}

\begin{algorithmic}
\State (1) Input:  initial approximation $u^0$, maximum number of iterations $k_{max}$, tolerance $\kappa >0$, and hyperparameters 
$\varepsilon >0$, $\lambda >1$, $\rho > 0$, and $\zeta\in\left(0,1\right)$. \\

\State (2) Set $k = 0$ and compute the solution $x^{(0)} = (S^{(0)},I^{(0)},R^{(0)},S^{*,(0)}, I^{*,(0)},R^{*,(0)})$ to \eqref{eq:SIRwithCompliance} with control $u=u^{(0)}$. \\

\State (3) Perform the following iteration.\\

\Repeat\\
    \State (i) Compute the costate dynamics $p^{(k)}$ obeying \eqref{eq:costateDyn} with $x=x^{(k)}$ and $u = u^{(k)}$.\\

   \State (ii) Compute $u^{(k+1)}$ via:
    $$u^{(k+1)}(t) = \argmin_{w \in U_{\text{ad}}}\, \mathcal H_\eps \left(x^{(k)}(t), p^{(k)}(t), u^{(k)}(t),w\right), \,\,\,\,\,\,\, t \in [0,T].$$    

    \State (iii) Compute the solution $x^{(k+1)}$ to the \eqref{eq:SIRwithCompliance} with control $u=u^{(k+1)}.$\\
    
    \State (iv) Compute $\tau=\|u^{(k+1)} -u^{(k)}\|^2_{(L^{2}[0,T])^4}$.\\

    \State (v) {\bf if:} $\left\{J\left(x^{(k+1)},u^{(k+1)}\right)-J\left(x^{(k)},u^{(k)}\right) \leq -
    \rho \, \tau \right\}$ 
    then 
    set $\eps \leftarrow \zeta \eps$ and continue to step (vi),\\
    \hspace{\algorithmicindent}\hspace{0.6cm}{\bf else if:} 
    $\left\{J\left(x^{(k+1)},u^{(k+1)}\right)-J\left(x^{(k)},u^{(k)}\right)> -\rho \, \tau\right\}$ then set $\eps \leftarrow \lambda\eps$ and go to step (ii). \\

    \State (vi) Set $k\leftarrow k+1.$\\

\Until{$k = k_{\text{max}}$ or $\tau < \kappa$}\\
\Return $(x^*,u^*)$ - the most recently computed state and control trajectories
\end{algorithmic}
\label{alg:SQHmethod}
\end{algorithm}

\WW{This is the key difference between SQH and the popular forward-backward sweep method (FBSM) \cite[Ch. 4]{Lenhart}. FBSM applies a simple fixed-point iteration to the system \eqref{eq:stateDyn}-\eqref{eq:OptCondition} and is exceedingly easy to implement, but lacks a robust convergence theory. SQH employs a similar forward-backward sweep methodology, but achieves a more stable and well-behaved iteration by using the quadratic perturbation to the Hamiltonian and then adjusting the perturbation coefficient $\eps$ to ensure sufficient decrease of the cost functional. This works by replacing the rigid application of the optimality condition \eqref{eq:OptCondition} with a much gentler application of a proximal step \cite[Ch. 6]{beck}, which allows us to prove that algorithm~\ref{alg:SQHmethod} does indeed terminate. We do so after establishing the following lemma. Due to the better theoretical properties and similar ease of implementation, we would argue that SQH is preferable for our problem, and indeed for more general problems in epidemic controls as seen in \cite{SQHepidemiology}.}

\begin{lemma}\label{lem:JHrelation}
    Assume that the running cost function $r$ from \eqref{eq:costFunc} has a Lipschitz continuous gradient. Let $u,v\in\mathscr U_{\text{ad}}$ be two control maps with associated solutions $x = \mathcal S(u)$ and $y = \mathcal S(v)$ of \eqref{eq:SIRwithCompliance}. Then there is a constant $c >0$ such that \begin{equation} \label{eq:B0} J(x,u) - J(y,v) \le \int^T_0 \Big (\mathcal H(x(t),p(t),u(t)) - \mathcal H(x(t),p(t),v(t))\Big)dt + c\|u-v\|^2_{(L^2[0,T])^4}\end{equation} where $p$ is the co-state trajectory corresponding to $(x,u)$.
\end{lemma}

\begin{proof}
    We compute, dropping dependence on $t$ for brevity, \begin{equation}\label{eq:B1} \begin{aligned}
        J(x,u)-J(y,v) &= \int^T_0 \Big( r(x)+\ell(u) - r(y) - \ell(v)\Big)dt \\
        &= \int^T_0 \Big( [\mathcal H(x,p,u) - \mathcal H(y,p,v)] + \innerprod{p}{f(y,v)-f(x,u)}\Big) dt \\
        &= \int^T_0 [\mathcal H(x,p,u) - \mathcal H(x,p,v)] dt \\ &\hspace{2cm}+ \int^T_0 \Big( [\mathcal H(x,p,v) - \mathcal H(y,p,v)]+\innerprod{p}{\tfrac d {dt}(y-x)}\Big)dt.
        \end{aligned}
    \end{equation} In the last equality of \eqref{eq:B1}, the first term on the right hand side is as it appears in \eqref{eq:B0}, so we work on the last line, which we call $I$. We see \begin{equation} \label{eq:B2}
        \begin{aligned}
            I &= \int^T_0 \Big( [\mathcal H(x,p,v) - \mathcal H(y,p,v)]+\innerprod{p}{\tfrac d {dt}(y-x)}\Big)dt \\ 
            &= \int^T_0 \left( \int^1_0 -\innerprod{\nabla_x \mathcal H(x+\lambda(y-x),p,v)}{y-x} d\lambda + \innerprod{p}{\tfrac d {dt}(y-x)}\right)  dt.
        \end{aligned}
    \end{equation} Now integrating by parts to pass the derivative onto $p$, the boundary terms vanish since $x(0)=y(0)$ and $p(T) = 0$. So using \eqref{eq:costateDyn}, \eqref{eq:B2} becomes \begin{equation} \label{eq:B3}I = \int^T_0 \int^1_0 \innerprod{\nabla_x \mathcal H(x,p,u) - \nabla_x H(x+\lambda(y-x),p,v)}{y-x} d\lambda dt.\end{equation} Now $\nabla_x \mathcal H = (J_f(x,u))^Tp+\nabla r$ where $J_f$ is the Jacobian matrix of $f$. Since $f$ is smooth (and its arguments remain bounded) and since $\nabla r$ is Lipschitz by assumption, $\nabla_x\mathcal H$ is Lipschitz. Thus \eqref{eq:B3} implies \begin{equation} \label{eq:B4}\begin{aligned}I &\le L \int^T_0 \int^1_0 (\lambda \abs{x-y}^2_2 + \abs{u-v}_2\abs{x-y}_2 )d\lambda dt \\ &\le \frac L 2 \|x-y\|^2_{(L^2[0,T])^6} + L\|x-y\|_{(L^2[0,T])^6}\|u-v\|_{(L^2[0,T])^4}.\end{aligned}\end{equation}  Finally, lemma~\ref{lem:CSlip} tells us that the control-to-state map is Lipschitz, whereupon \eqref{eq:B4} yields $$I \le c\|u-v\|^2_{(L^2[0,T])^2}$$ and inserting this last bound into \eqref{eq:B1} concludes the proof. 
\end{proof}

\begin{theorem} \label{thm:theoremSQH1} Under the assumptions of lemma~\ref{lem:JHrelation}, let $\left(x^{(k)},u^{(k)}\right)$ and $\left(x^{(k+1)},u^{(k+1)}\right)$ be successive iterates generated by the SQH method (algorithm \ref{alg:SQHmethod}). Then there is $c > 0$ such that
\begin{equation} \label{eq:decreaseJ} J\left(x^{(k+1)},u^{(k+1)}\right) - J\left(x^{(k)},u^{(k)}\right) \le  -(\eps-c) \, \| u^{(k+1)} - u^{(k)} \|^{2}_{(L^{2}[0,T])^4}.\end{equation} for the current value of $\eps > 0$ in algorithm \ref{alg:SQHmethod}.
In particular, this implies \begin{equation}\label{eq:decreaseJ2}J\left(x^{(k+1)},u^{(k+1)}\right) - J\left(x^{(k)},u^{(k)}\right) \le - \rho \, \tau\end{equation} for $\eps \ge \rho+c$ and $\tau = \| u^{k+1} - u^k \|^{2}_{(L^{2}[0,T])^3}$. \end{theorem}

\begin{proof}
    In step (ii), we resolve $u^{(k+1)}$ such that $$\mathcal H_\eps \left(x^{(k)}(t), p^{(k)}(t), u^{(k)}(t),u^{(k+1)}(t)\right) \le \mathcal H_\eps \left(x^{(k)}(t), p^{(k)}(t), u^{(k)}(t),w\right),$$ for all $w \in U_{\text{ad}} \text{ and } t \in [0,T].$ In particular, choosing $w = u^{(k)}(t)$ at each $t \in [0,T]$ (and then dropping the argument $t$ for brevity), we see $$\mathcal H_\eps \left(x^{(k)}, p^{(k)}, u^{(k)},u^{(k+1)}\right) \le \mathcal H_\eps \left(x^{(k)}, p^{(k)}, u^{(k)},u^{(k)}\right) = \mathcal H\Big(x^{(k)},p^{(k)},u^{(k)}\Big).$$ Inserting the definition of $\mathcal H_\eps$, this yields $$\mathcal H\Big(x^{(k)},p^{(k)},u^{(k+1)}\Big) -\mathcal H\Big(x^{(k)},p^{(k)},u^{(k)}\Big) \le  -\eps \big\vert \pow u {k+1} - \pow u k\big\vert ^2_2.$$ Then from lemma~\ref{lem:JHrelation}, we have \begin{align*}
        J\Big(\pow x{k+1}, \pow u{k+1}\Big) - J\Big(\pow x k, \pow u k\Big) &\le \int^T_0 \Big(H\Big(x^{(k)},p^{(k)},u^{(k+1)}\Big) -\mathcal H\Big(x^{(k)},p^{(k)},u^{(k)}\Big) \Big)dt \\ &\hspace{3cm}+ c\|\pow u {k+1} - \pow u k\|^2_{(L^2[0,T])^4} \\
        &\le -(\eps - c) \|\pow u {k+1} - \pow u k\|^2_{(L^2[0,T])^4}
    \end{align*} as desired. 
\end{proof}

Recall, if sufficient decrease of the cost functional is not achieved, $\eps$ is successively increased. Theorem~\ref{thm:theoremSQH1} tells us that once $\eps \ge \rho + c$, we are guaranteed to achieve sufficient decrease of the cost functional. Given this, we prove that the numerical iterates do indeed converge. 

\begin{theorem} \label{thm:themremSQH2}
    The sequence of iterates $\{(\pow x k, \pow u k)\}$ created by the SQH algorithm satisfy the following. \begin{itemize}
        \item[(a)] $\{J\Big(\pow x k,\pow uk\Big)\}$ decreases in $k$, and thus converges to some $J^* \ge \min_{u\in\mathscr U_{\text{ad}}} J(x,u)$ 
        \item[(b)] $\displaystyle \|\pow u {k+1}-\pow u k\|_{(L^2[0,T])^4}, \,\|\pow x {k+1}-\pow x k\|_{(L^2[0,T])^6} \to 0$ as $k \to \infty$
    \end{itemize}  
\end{theorem}

\begin{proof}
    Monotonicity of $J(\pow x k, \pow u k)$ follows directly from \eqref{eq:decreaseJ2} following the observation above that eventually $\eps$ will be increased enough that $\eps \ge \rho +c$. Convergence then follows by the monotone convergence theorem, so (a) is proven. 

    For (b), we rearrange \eqref{eq:decreaseJ2} to see that when $\pow u{k+1}$ is resolved, $$\|\pow u{k+1}-\pow u k\|^2_{(L^2[0,T])^4} \le \frac{1}{\rho} \Big(J\Big(\pow x k, \pow u k\Big) - J\Big(\pow x {k+1},\pow u{k+1}\Big)\Big).$$ Thus we have the telescoping sum $$\sum^K_{k=0} \|\pow u{k+1}-\pow u k\|^2_{(L^2[0,T])^4} \le \frac 1 \rho \Big(J\Big(\pow x 0, \pow u 0\Big) - J\Big(\pow x {K+1},\pow u{K+1}\Big)\Big).$$ As $K\to \infty$ the right hand side above converges, so this shows that the left hand side is a convergent series of positive terms, implying that the summands go to zero. That $\|\pow x {k+1}-\pow x k\|_{(L^2[0,T])^6} \to 0$ as well follows by the Lipschitz continuity of the control-to-state map proven in lemma~\ref{lem:CSlip}.
\end{proof}

With this, we move on to simulations and discussion wherein we define a cost functional and discuss the optimal control strategies that result in the case of different policy-maker preferences.

\section{Results \& Discussion} \label{sec:results}

We implement algorithm \ref{alg:SQHmethod} in MATLAB and demonstrate the results in several different cases. Specifically, we consider the cost functional \eqref{eq:costFunc} with \begin{equation}
    \label{eq:rl}
    \begin{aligned}
        r(x) &= c_1(I+I^*) + c_2 (S^*+I^*+R^*), \\
        \ell(u) &= d_1(\abs{\alpha} +\abs{\eta}+\abs{\mu}+\abs{\nu}) + \frac{d_2}2(\alpha^2 + \eta^2+\mu^2+\nu^2).
    \end{aligned}
\end{equation} where $c_1,c_2,d_1,d_2$ are nonnegative weights that express the policy-maker's preferences. We are interested in observing the differing behavior of optimal control maps in response to changes in the policy-maker's preferences. That is, in each scenario below, we fix these weights and use the SQH method to resolve approximate optimal control maps. \WW{The hyperparameters used for the SQH algorithm are  $\kappa = 10^{-8}, \lambda = 1.1, \zeta = 0.9, \rho = 10^{-9}$ and an initial value of $\eps = 1$. We use a maximum iteration count of $k_{\text{max}}=10000$ but the algorithm terminated before reaching this count in all of our simulations (typically requiring on the order of a few hundred iterations).}

For all simulations we set $T = 100$, and for the sake of numerically resolving the solutions of \eqref{eq:SIRwithCompliance} and \eqref{eq:costateDyn}, we use an explicit Euler scheme with uniform discretization of $[0,T]$ with step size $\Delta t = 0.1$. 

What remains is to specify initial conditions and the parameters $b,\delta,\beta,\gamma,\overline \eta,\xi,\mu^*,\overline \mu,\overline \nu$. To enhance readability, we include a description of each of these parameters and the values they take in each scenario. In all scenarios, we take $b = \delta$, and initial condition $(S_0,I_0,R_0,S^*_0,I^*_0,R^*_0) = (0.69,0.01,0,0.29,0.01,0)$, indicating that at the outset of the epidemic, 2\% of the population is infected, and 30\% of the population is noncompliant. Because $b =\delta$ and the initial condition has total population $1$, by \eqref{eq:TOTALPOP}, the total population will be $1$ for all time, so individual compartments can be interpreted as percentages of the total population.  

\WW{Before displaying results, we briefly discuss the bearing of the controls. Recall that the cost \eqref{eq:costFunc} for the policy-maker is split into three parts: (C1) cost associated with infections (weighted by the constant $c_1$), (C2) costs associated with noncompliance (weighted by the constant $c_2$) and (C3) cost associated with controls (weighted by constants $d_1,d_2$). By proposition \ref{prop:sanity}, any use of controls will help decrease cost (C1) by slowing the spread of infections. This is the only effect that $\alpha$ and $\eta$ can have. The controls $\mu$ and $\nu$ obliquely lower infection rate and thus reduce cost (C1), but more explicitly decrease cost (C2) by decreasing the noncompliant population. Of course, any use of control increases cost (C3), meaning that the policy-maker has a trade-off to consider.} We mention specifically that the absolute size of the weights is not important, but rather their size relative to each other. Accordingly, in much of the below we fix the control cost weights $d_1 = d_2 = 0.1$ and vary the cost weights $c_1,c_2$ associated with total infections and total amount of noncompliance, respectively. In certain cases, we would like to ``turn off" individual controls; to do so, one can simply set the upper bound for that control to zero (this could equivalently be accomplished by sending the control's cost weight to infinity).  Lastly, as baseline values for infection and noncompliance cost weights, we choose $c_1 = 1$ and $c_2 = 0.1$. This causes the cost due to infection and the cost due to noncompliance to be roughly equal in the uncontrolled case. The reason that $c_2$ must be significantly smaller is that with no controls, noncompliance is a permanent state, whereas infection is not, which means that noncompliance has the potential to contribute much more to the total cost. Finally, we emphasize that for each separate plot of controls below, we are changing certain parameters and resolving optimal controls again using the SQH method (i.e., none of these plots contain synthetically designed controls; all of the controls presented are solutions of an optimal control problem with different parameters and/or cost weights). \WW{The MATLAB code which runs these simulations is available at the link below.\footnote{\href{https://github.com/chparkinson/Opt_Control_SIR_With_Noncompliance_As_Social_Contagion}{github.com/chparkinson/Opt\_Control\_SIR\_With\_Noncompliance\_As\_Social\_Contagion}}}

\begin{figure}[t!]
\begin{tabular}{|lcrrr|}
\hline
    Parameter & Description & Scen. 1 & Scen. 2 & Scen. 3\\ \hline
    $b$ & natural birth rate & 0.01 & 0.01 & 0.01\\
    $\delta$ & natural death rate & 0.01 & 0.01 & 0.01\\
    $\beta$ & disease infection rate & 0.4 & 0.6 & 0.4\\
    $\gamma$ & disease recovery rate & 0.2 & 0.2 & 0.2\\
    $\overline \alpha$ & max. achievable decrease in mixing due to compliance & 0.5 & 0.5 & 0.5\\
    $\overline \eta$ & maximum achievable treatment efficacy & 0.1 & 0.1 &0.1\\
    $\xi$ & noncompliant portion of new population & 0 & 0.3 & var.\\
    $\mu^*$ & baseline infection rate for noncompliance & 0.1 & 0.2 & 0.3\\
    $\overline \mu$ & max. achievable decrease in inf. rate of noncomp. & 0.05 & 0.1 & 0.15\\
    $\overline \nu$ & maximum achievable rehab. rate for noncompliance & 0.1 & 0.1 & 0.1\\
    $c_1$ & cost weight on infections & var. & var. & 1 \\
    $c_2$ & cost weight on noncompliance & var. & 0.01 & 0.05 \\
    $d_1$ & $L^1$ cost weight on controls & 0.1 & 0.1 & 0.1\\
    $d_2$ & $L^2$ cost weight on controls & var. & 0.1 & 0.1 \\
    \hline
\end{tabular}
\caption{Table of parameters for each scenario. Recall in particular that $\overline \alpha, \overline \eta, \overline \mu, \overline \nu$ are upper bounds for the control variables $\alpha,\eta,\mu,\nu$, which represent decrease in disease spread rate due to NPIs, increase in recovery rate due to seeking treatment, decrease in spread rate of noncompliance due to public health initiatives, and recovery rate for noncompliance due to educational campaigns, respectively.}
\label{fig:params}
\end{figure}

\subsection{Scenario 1} The parameters for this scenario are listed in figure \ref{fig:params}, and the result of the simulation in the absence of controls is included in figure \ref{fig:Scenario1_Uncontrolled}. With no controls, the reproductive ratio for this scenario is $\mathscr R_0 = \tfrac b \delta \cdot \frac{\beta}{\gamma+\delta} = \frac{0.4}{0.21} \approx 1.9.$ We note in particular that $\xi = 0$, meaning that all newly introduced members of the population are compliant, and we are in the regime of theorem \ref{thm:DFEStability}(i),(ii). In this case, there are two DFEs, that, when controls are constant, take the forms $(s_1,s^*_1) = (\tfrac b \delta,0)$ and $(s_2,s_2^*) = (\tfrac{\nu+ \delta}{\mu^*-\mu }, \tfrac b \delta - \tfrac{\nu+\delta}{\mu^*-\mu})$. If the controls are constant, the second is stable in the case that it is physically meaningful (i.e. $\tfrac b \delta < \tfrac{\nu+\delta}{\mu^*-\mu}$); otherwise, the first is stable.  Clearly to minimize (C1) and (C2), it behooves the policy-maker to drive the system toward $(s_1,s_1^*) = (\tfrac b \delta,0)$. If control cost is no issue, the policy-maker can increase $\mu,\nu$ to ensure that $(s_1,s_1^*)$ is stable, and then decrease $\mathscr R_0 (\tfrac b \delta,0) = \tfrac b \delta \cdot \tfrac{\beta(1-\alpha)}{\gamma + \eta + \delta}$ using $\alpha, \eta$ so as to fall into the stability regime of theorem \ref{thm:DFEStability}(i). However, it will likely be too expensive in terms of (C3) to actually fully do this. In figure \ref{fig:Scenario1_balanced}, we see the optimal control plans for a ``balanced" policymaker: one who would like to minimize both infections and noncompliance. The disease dynamics are plotted in the top panels of figure \ref{fig:Scenario1_balanced}, split into compliant and noncompliant compartments. The bottom left panel of figure \ref{fig:Scenario1_balanced} displays the optimal control plans. Finally, the bottom right displays the reproductive ratio $\mathscr R_0(s,s^*)$ which, due to the controls, now changes with time. In particular, we plot the value of $\mathscr R_0(\tfrac b \delta,0)$ when $\tfrac b \delta < \frac{\nu + \delta}{\mu^* - \mu}$ to reflect the conditions of theorem \ref{thm:DFEStability}(i) and $\mathscr R_0(\tfrac{\nu+ \delta}{\mu^*-\mu }, \tfrac b \delta - \tfrac{\nu+\delta}{\mu^*-\mu})$ when $\tfrac b \delta \ge \frac{\nu + \delta}{\mu^* - \mu}$ to reflect the conditions of theorem \ref{thm:DFEStability}(ii). In either case, if the corresponding reproductive ratio was kept below 1, the corresponding DFE would be stable, but from the plot we see that it is not optimal to do so. Instead, controls $\mu$ and $\nu$ are used to ensure compliance in the early stages of the epidemic, so that $\alpha$ and $\eta$ can be effectively deployed to keep the reproductive ratio low, though a small outbreak is indeed allowed to occur. With these cost weights, the optimal controls reduce the total cost from $8.47$ in the uncontrolled case to $4.79$ in the optimally controlled case, representing a roughly $43.5\%$ relative reduction in cost due to controls. 

\begin{figure}[t!]
    \centering
    \includegraphics[width=0.72\linewidth]{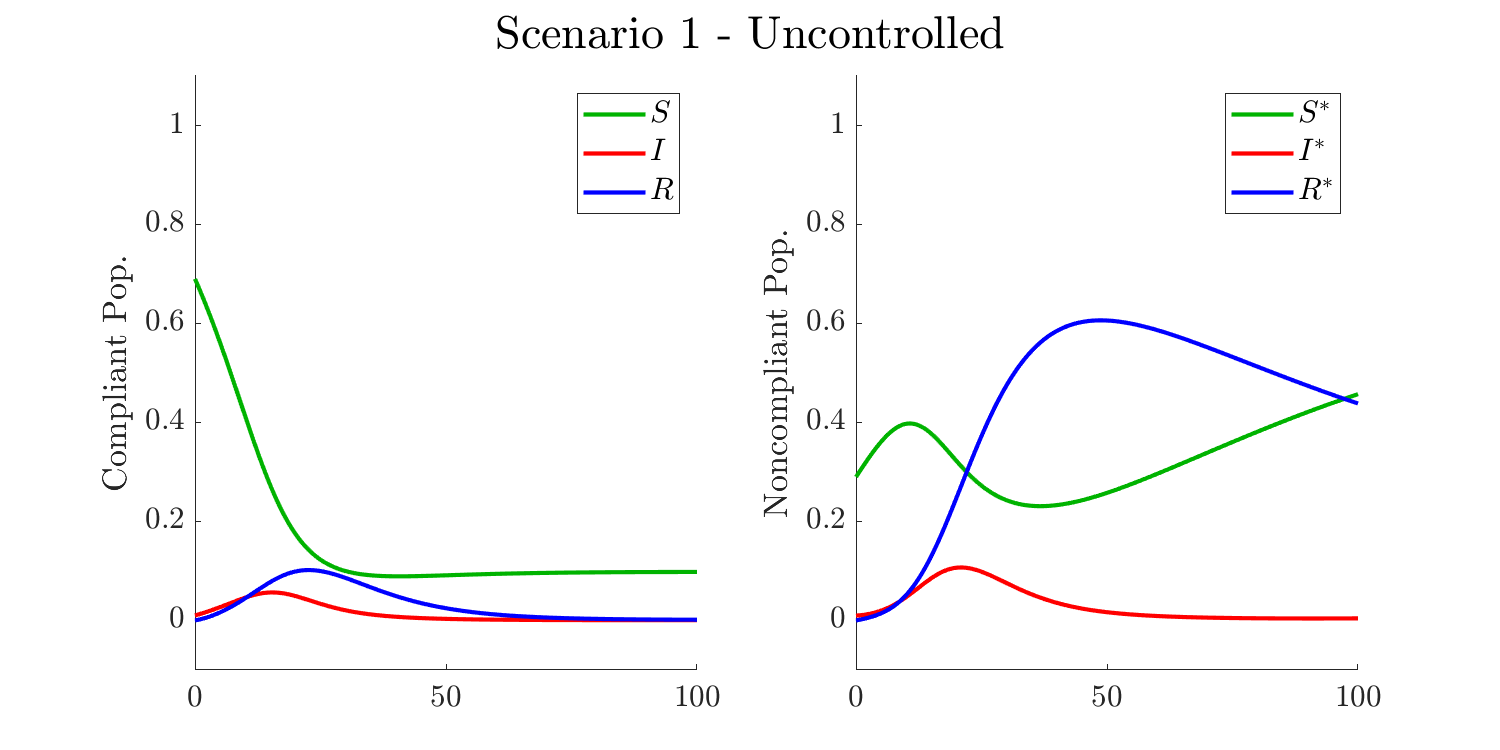}
    \caption{Simulation of Scenario 1 in the absence of controls.}
    \label{fig:Scenario1_Uncontrolled}
\end{figure}

\begin{figure}[t!]
\centering
\includegraphics[width=0.72\linewidth]{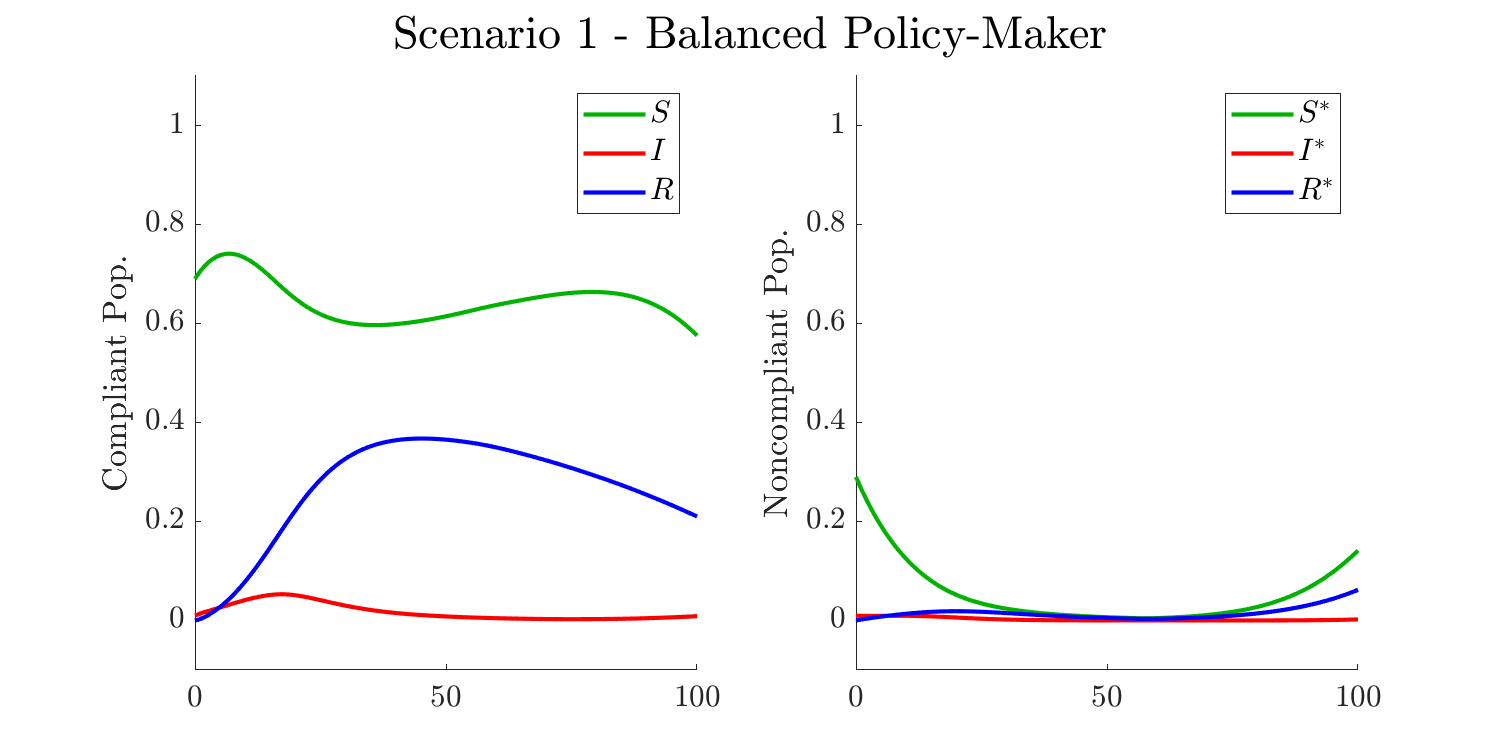}\\
\includegraphics[width=0.72\linewidth]{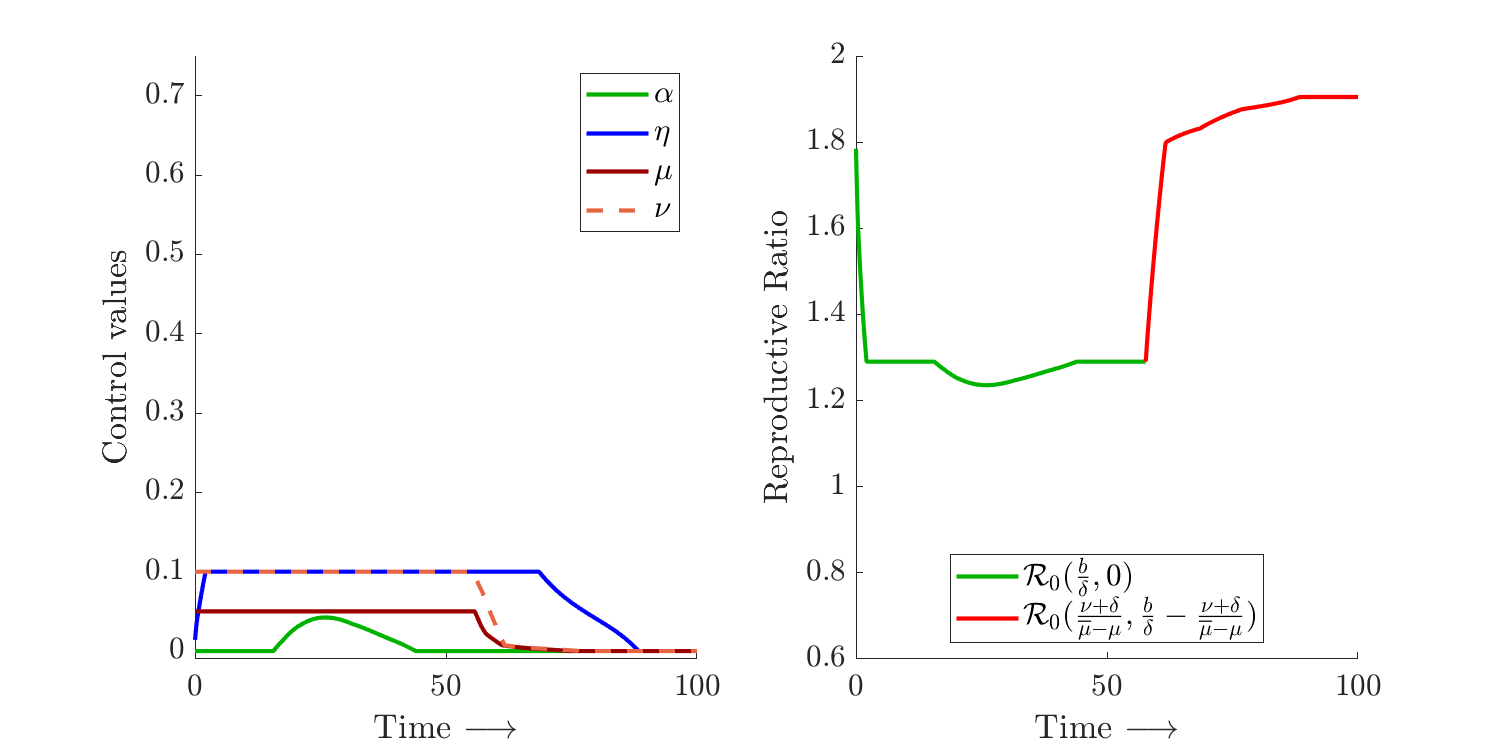}
    \caption{Simulation of Scenario 1 with a ``balanced" policy-maker: cost weights $c_1 = 1, c_2 = 0.05$, $d_1 = 0.1, d_2 = 0.1$. The strategy is to quell noncompliance using $\mu,\nu$ so as to keep the infection rate low. Since the majority of the population is compliant, $\alpha$ and $\nu$ can be deployed to effectively lower the reproductive ratio. However, it is not lowered so much that $\mathscr R_0(\tfrac b \delta,0) < 1.$ A small outbreak of the disease is allowed to occur. } 
    \label{fig:Scenario1_balanced}
\end{figure}

\begin{figure}[t!]
\centering
\includegraphics[width=0.72\linewidth]{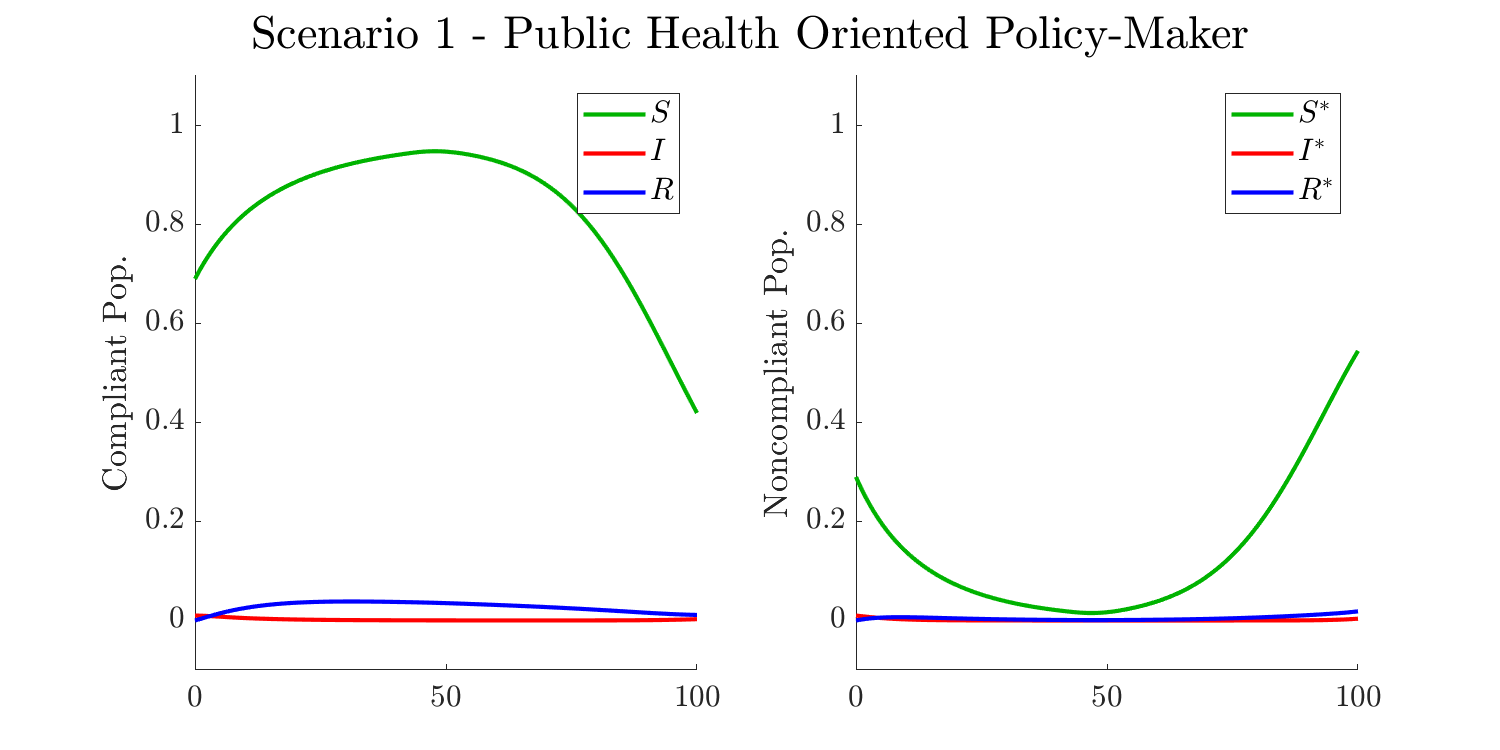}\\
\includegraphics[width=0.72\linewidth]{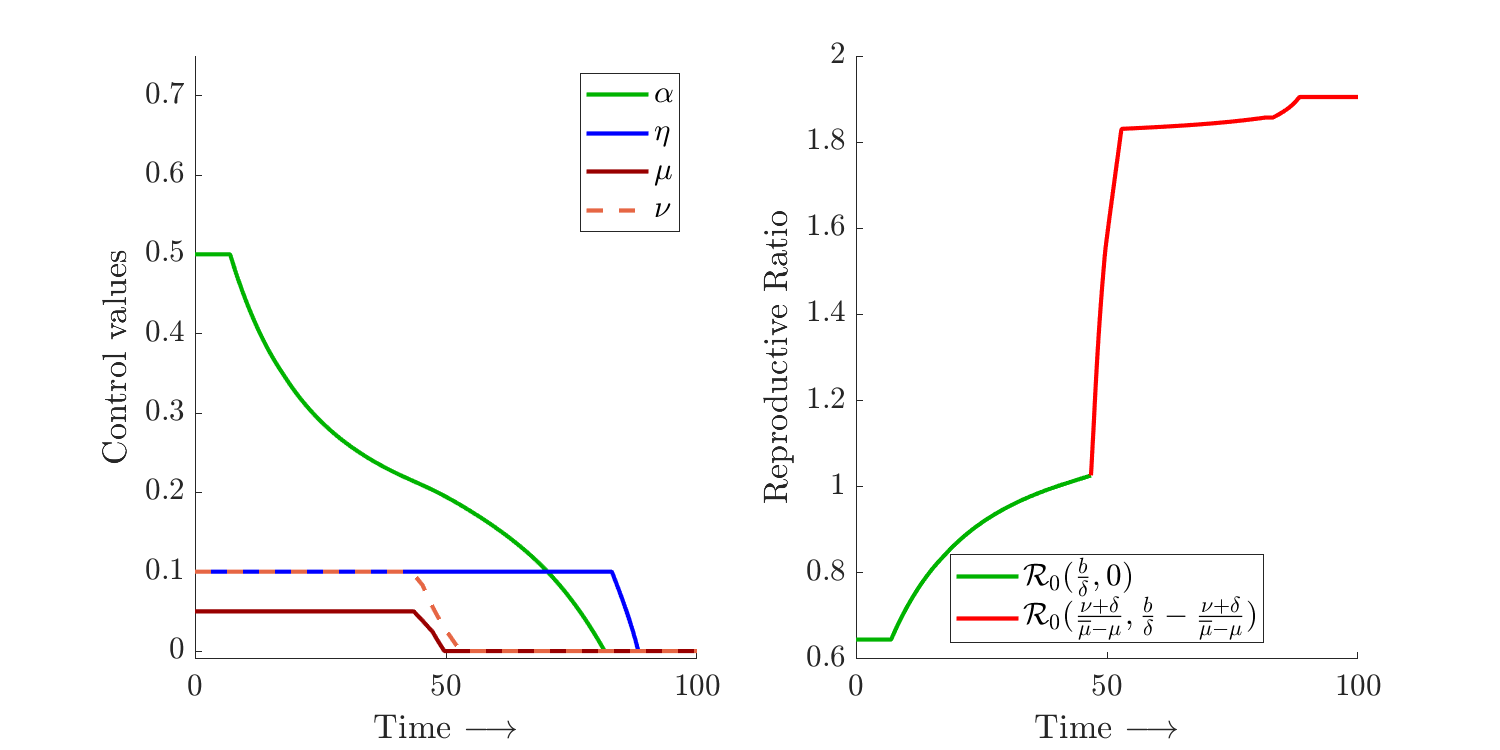}
    \caption{Simulation of Scenario 1 with a public health oriented policy-maker: cost weights $c_1 = 5, c_2 = 0, d_1 = 0.1, d_2 = 0.1$. In this case, even though there is no explicit cost associated with noncompliance, because the policy-maker desires to keep infections low, the optimal strategy employs $\mu$ and $\nu$ even more strongly than in the balanced case. The disease controls $\alpha$ and $\eta$ are also increased so as to drive $\mathscr R_0(\tfrac b \delta,0) < 1$ and take advantage of the stability of the DFE $(\tfrac b \delta,0)$ until the disease has been effectively eradicated.} 
    \label{fig:Scenario1_healthy}
\end{figure}

We contrast this with the optimal actions corresponding to cost weights $c_1 = 5, c_2 = 0$. This represents a public health oriented policy-maker whose primary concern is to stop the disease, and does not even assess cost based on noncompliance. The results of the simulation using these weights are in figure \ref{fig:Scenario1_healthy}. In this case, the cost weight on infections is high enough that it is indeed optimal to use controls $\mu,\nu$ to ensure that, in the initial stages of the epidemic, we remain in the stability regime for the disease-free equilibrium $(s_1,s_1^*) = (\tfrac b \delta,0)$, and then use $\alpha,\eta$ to drive $\mathscr R_0(\tfrac b \delta,0) < 1$, in essence taking advantage of theorem \ref{thm:DFEStability}(i). This is maintained until the disease has been effectively eradicated, at which time the policy-maker eases the controls (note that, in the context of this model, the policy-maker is aware of the horizon time $T = 100$, so $\alpha$ is tapered off such that infections will not significantly increase within this final time). Figure \ref{fig:Scenario1_healthy} demonstrates the need for synergy between the disease controls $\alpha,\eta$ and noncompliance controls $\mu,\nu$. With these cost weights, there is no explicit cost associated with increased noncompliance. However, because increased noncompliance accelerates disease spread, the policy-maker must still employ $\mu,\nu$ to keep noncompliance low so as to accomplish their goal of disease elimination. With these cost weights, the total cost is reduced from $22.60$ in the uncontrolled case to $8.62$ in the optimally controlled case, representing a roughly $61.9\%$ relative reduction in cost. 

To further explore the synergy between the controls, we include one final figure using the scenario 1 parameters, and the health oriented policy-maker with cost weights $c_1=5,c_2=0$. In this case, we demonstrate the effects of ``turning off" certain controls. The results are in figure \ref{fig:Scenario1_healthy_controlsoff}, where we display the effects on the optimal control maps of forcing $\eta = 0$ (top), or forcing $\mu = \nu = 0$ (bottom). Recall, $\eta$ represents treatment efforts that hasten recovery among the compliant infected population, so one could think of $\eta = 0$ as a scenario near the outset of an epidemic where treatment is unavailable or ineffective. In this case, the policy-maker's optimal strategy is to more strongly enforce NPIs (increase $\alpha$), while also using $\mu,\nu$ to ensure compliance. While they are not displayed, the dynamics in this case look, in essence, the same as those in the top panel of figure \ref{fig:Scenario1_healthy} in that the disease is effectively eradicated. However, without access to the control $\eta$, the total cost can only be reduced from $22.60$ in the uncontrolled case to $13.16$ in the optimally controlled case, representing only a roughly $41.8\%$ relative cost reduction (compared to a $61.9\%$ relative cost reduction when $\eta$ is available). Thus, in some sense, control efforts are only two-thirds as effective without $\eta$. 

\begin{figure}[t!]
\centering
\includegraphics[width=0.72\linewidth]{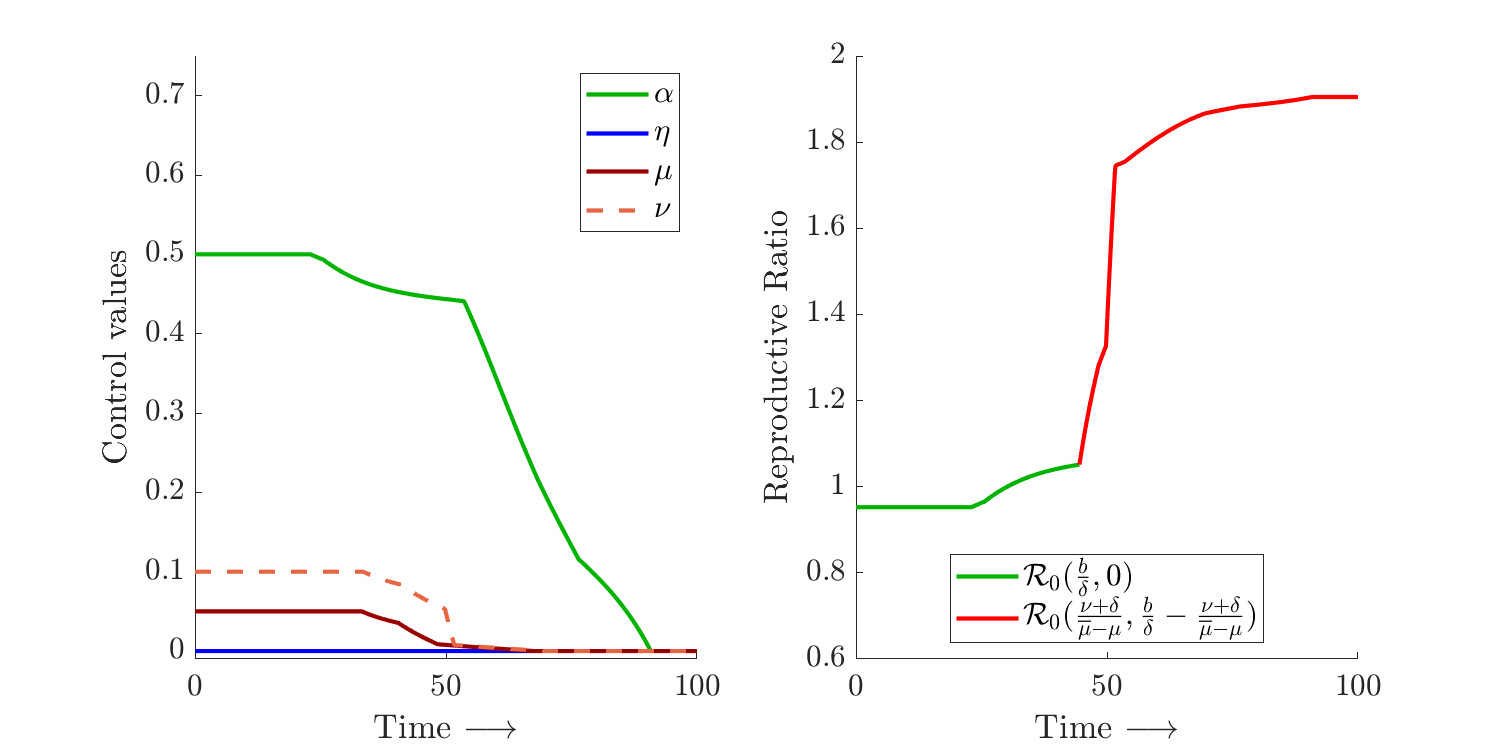}\\
\includegraphics[width=0.72\linewidth]{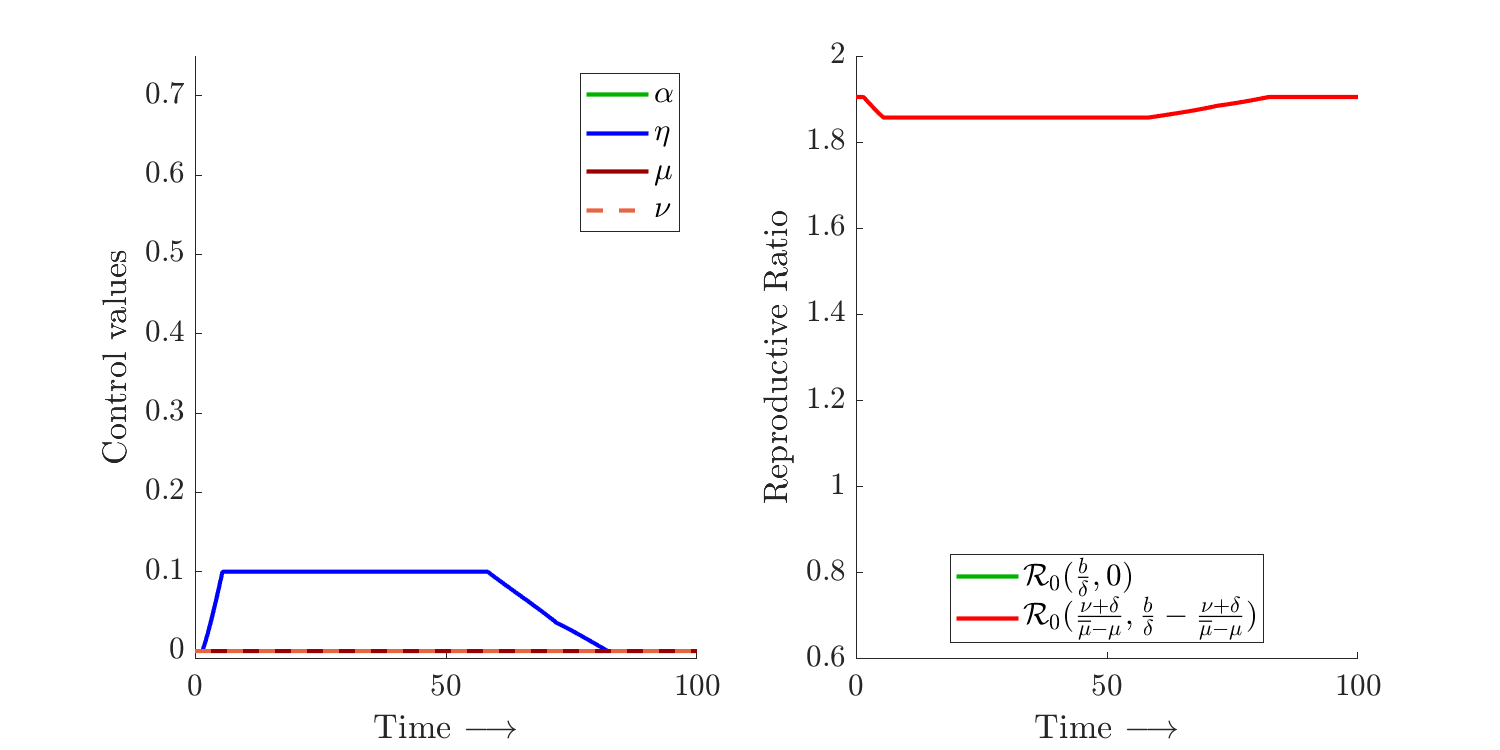}
    \caption{Simulation of Scenario 1 with a health oriented policy-maker: cost weights $c_1 = 5, c_2 = 0, d_1 = 0.1, d_2 = 0.1$, with either treatment $\eta$ unavailable (top), or control of noncompliance $\mu,\nu$ unavailable (bottom).} 
    \label{fig:Scenario1_healthy_controlsoff}
\end{figure}

In a similar vein, $\mu = \nu = 0$ represents a scenario wherein it is impossible to encourage compliance or slow the spread of noncompliance. In this case, the dynamics look in essence like those in figure \ref{fig:Scenario1_Uncontrolled} where no controls are used. One sees in the bottom panel of figure \ref{fig:Scenario1_healthy_controlsoff} that without the ability to curb noncompliance, it is still beneficial to slightly decrease the reproductive ratio among the compliant population using $\alpha$ and $\eta$, but much less so since the majority of the population will become noncompliant, hampering disease prevention efforts. With $\mu =\nu = 0$, one cannot land in the stability regime of theorem \ref{thm:DFEStability}(i), nor accomplish $\mathscr R_0(\frac{\nu+\delta}{\mu^* - \mu},\frac b \delta - \frac{\nu+\delta}{\mu^*-\mu})<1$, so the disease is simply allowed to run its course. The use of $\alpha$ and $\eta$ still accomplish a reduction of total cost from $22.60$ in the uncontrolled case to $20.31$ in the optimally controlled case, representing a roughly $10.1\%$ relative reduction in total cost (compared to the $61.9\%$ cost reduction when all controls are available). Meaning that, with $\mu=\nu=0$, control efforts are only about one-sixth as effective. 

Finally, for our last simulation in scenario 1, we demonstrate that our numerical method can handle the case of $L^1$ control costs by setting $d_2 = 0$, eliminating the $L^2$ costs from \eqref{eq:rl}. In this case, the Hamiltonian \eqref{eq:Hamiltonian} is affine in $\eta,\nu\mu$, while still being quadratic in $\alpha$, since the dynamics are quadratic in $\alpha$. Because an affine function on an interval maximized at one of the endpoints, in this case, $\eta,\mu,\nu$ should only ever take the value 0 or their maximum value; these are akin to bang-bang controls (though they could perhaps more accurately be called ``bang-no bang"). The control $\alpha$ does not have this structure and can still vary smoothly. Indeed, this is precisely what we see in figure~\ref{fig:Scenario1_balanced_L1}. This demonstrates that our methods are well-equipped to handle $L^1$ control problems. It also shows why allowing for controls in $BV([0,T];U_{\text{ad}})$ is far preferable to higher regularity spaces: controls which are discontinuous but have bounded variation readily appear in applications.

\begin{figure}[t!]
\centering
\includegraphics[width=0.72\linewidth]{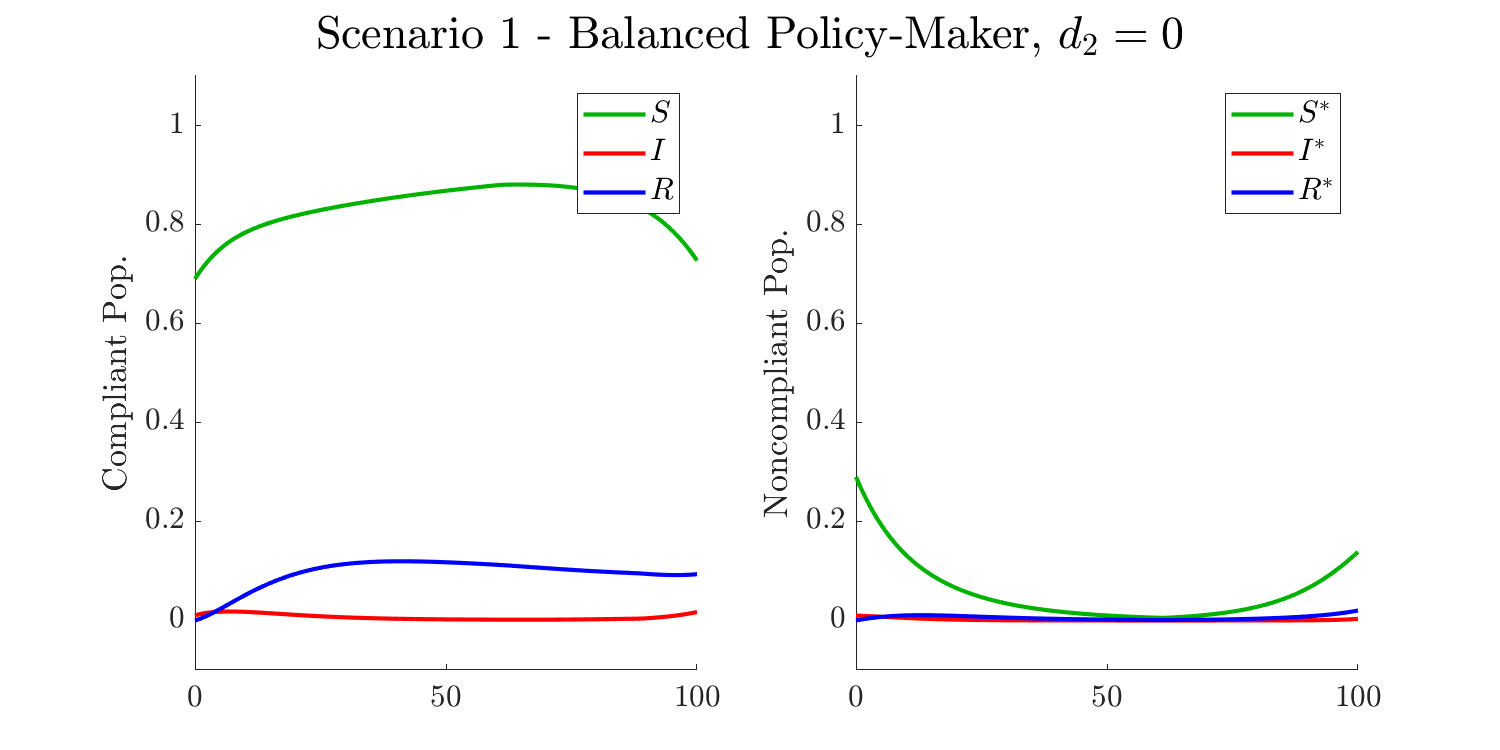}\\
\includegraphics[width=0.72\linewidth]{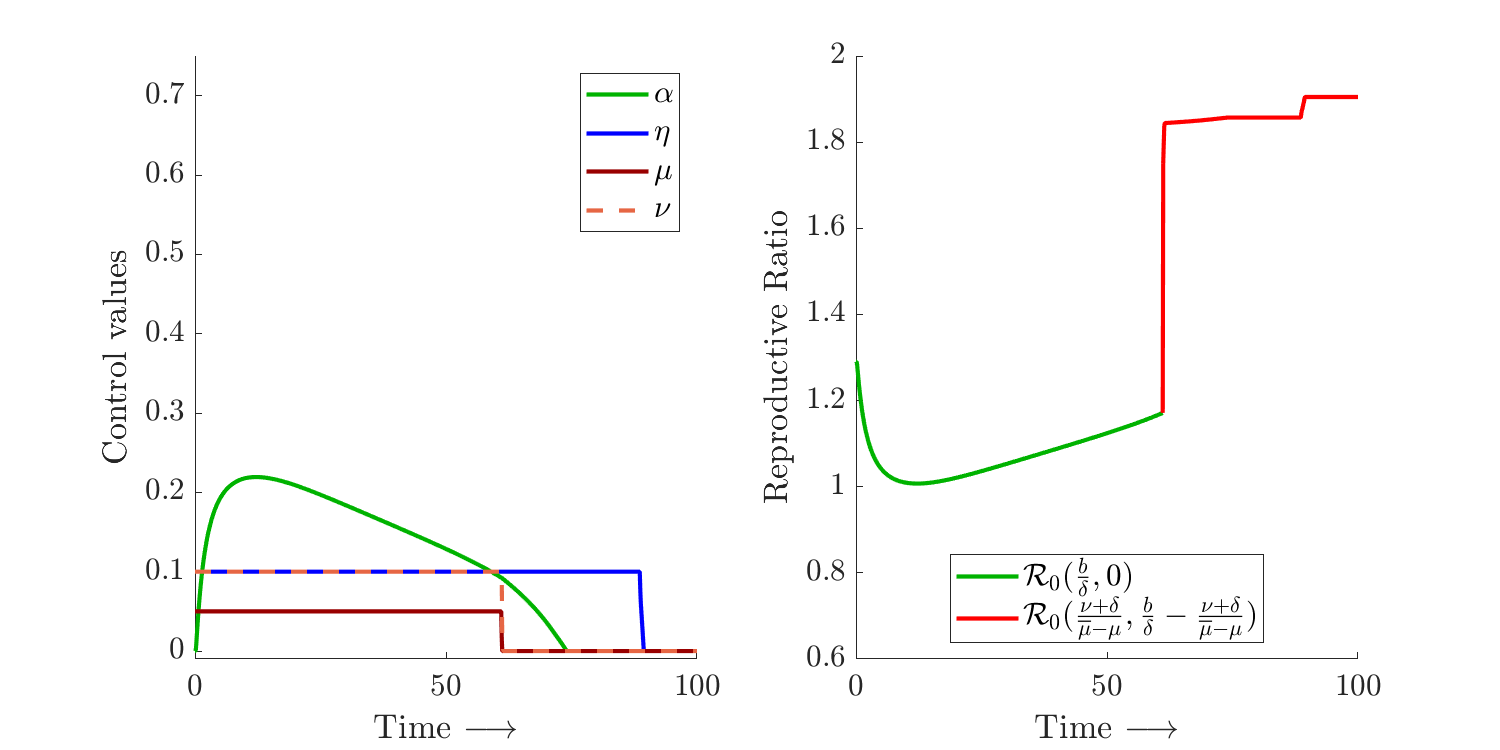}
    \caption{Simulation of Scenario 1 with a balanced policy-maker with no $L^2$ control costs: cost weights $c_1 = 1, c_2 = 0, d_1 = 0.1, d_2 = 0$. Theoretically, the controls $\eta,\mu,\nu$ should only take their maximum value or zero, and indeed this is borne out in the simulation. This demonstrates that our methods are well-equipped to handle $L^1$ control problems.} 
    \label{fig:Scenario1_balanced_L1}
\end{figure}

\subsection{Scenario 2} In this scenario, when compared with scenario 1, we increase the disease infection rate $\beta$ and the spread rate $\mu^*$ of noncompliance, and we also assume that $30\%$ of the newly introduced population members are noncompliant ($\xi = 0.3$), while all other parameters are unchanged. For this scenario, we set $c_2 = 0.01$, indicating a weak preference of the policy-maker to quell noncompliance for its own sake, though quelling noncompliance for the sake of slowing disease spread is still desirable. In figure \ref{fig:Scenario2}, we display the results of the simulation with $c_1$ successively taking values $1/3,1,3,6,9$. In this case, noncompliance is much more difficult to control since 30\% of the newly introduced population is noncompliant, and noncompliance spreads fairly quickly. When noncompliance is difficult to control, the control variables $\alpha$ and $\eta$ have a diminished effect. Accordingly, in this parameter regime, only a very strongly public-health oriented policy-maker will fully suppress disease spread. Note, for example, that the dynamics are qualitatively similar in the cases of $c_1 =1$ and $c_1 = 3$. It isn't until $c_1 = 6$ that the policy-maker significantly increases the use of the control $\alpha$ so as to more effectively stunt disease spread. This indicates that in scenarios where disease spread is fast enough, and noncompliance is rampant enough, a policy-maker who is even slightly economically minded (i.e. concerned with keeping control costs down) will let the disease run its course.

\begin{figure}[t!]
\centering
\includegraphics[width=0.85\linewidth]{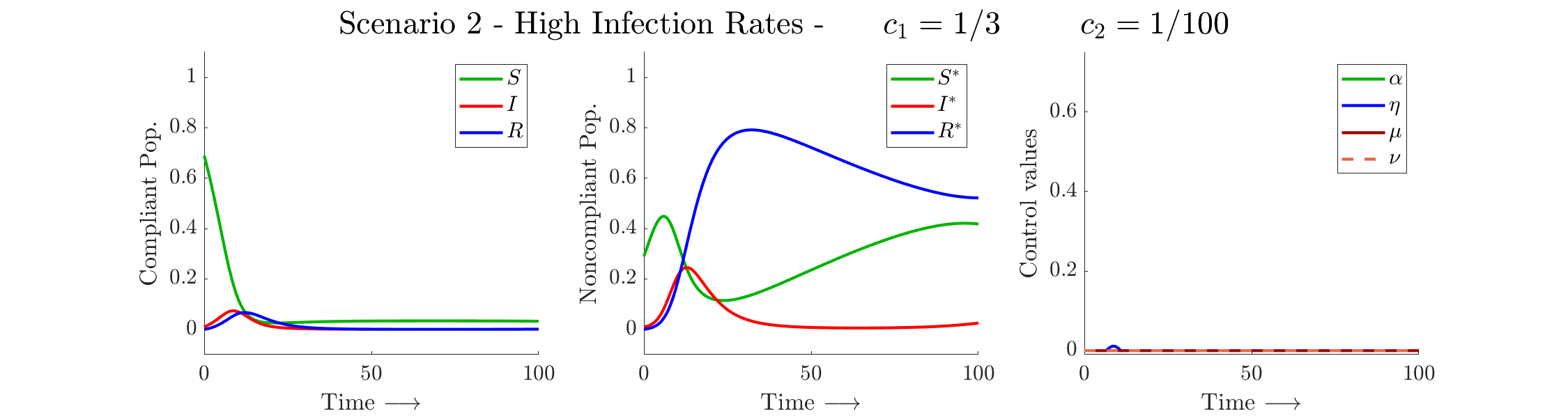}\\
\includegraphics[width=0.85\linewidth]{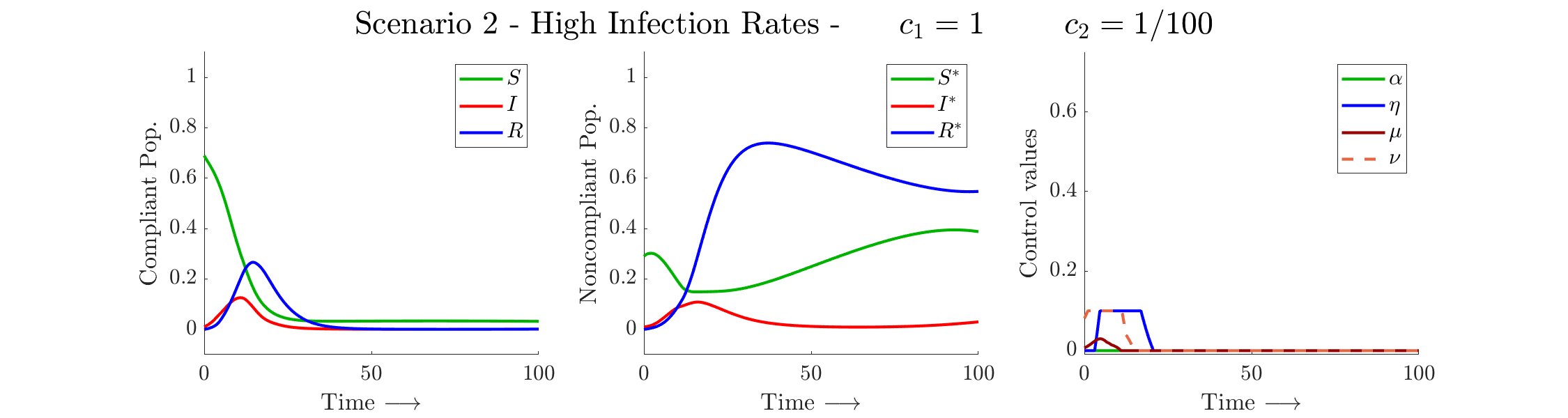}\\
\includegraphics[width=0.85\linewidth]{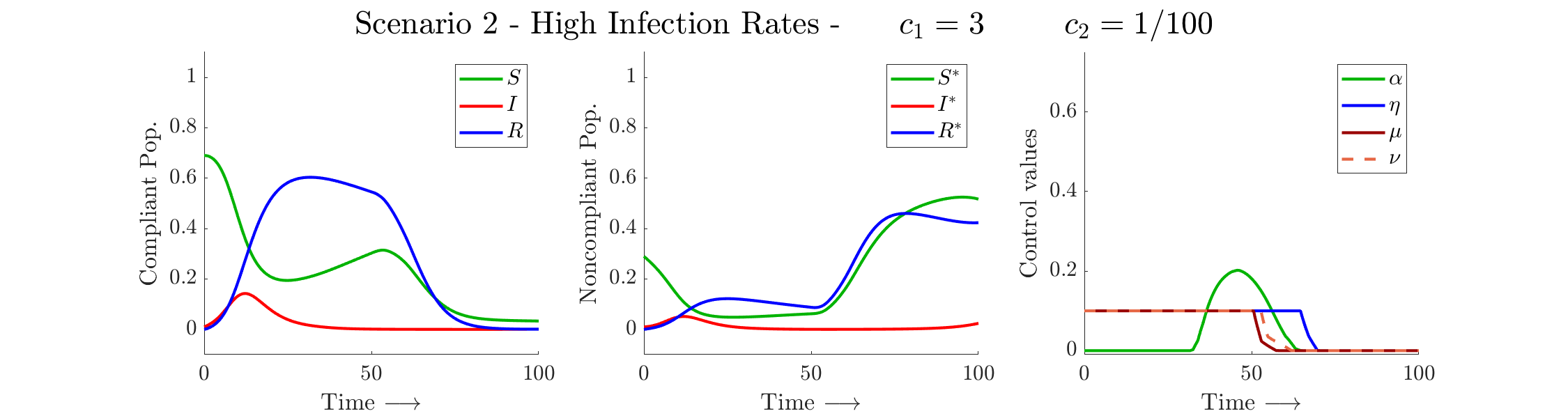}\\
\includegraphics[width=0.85\linewidth]{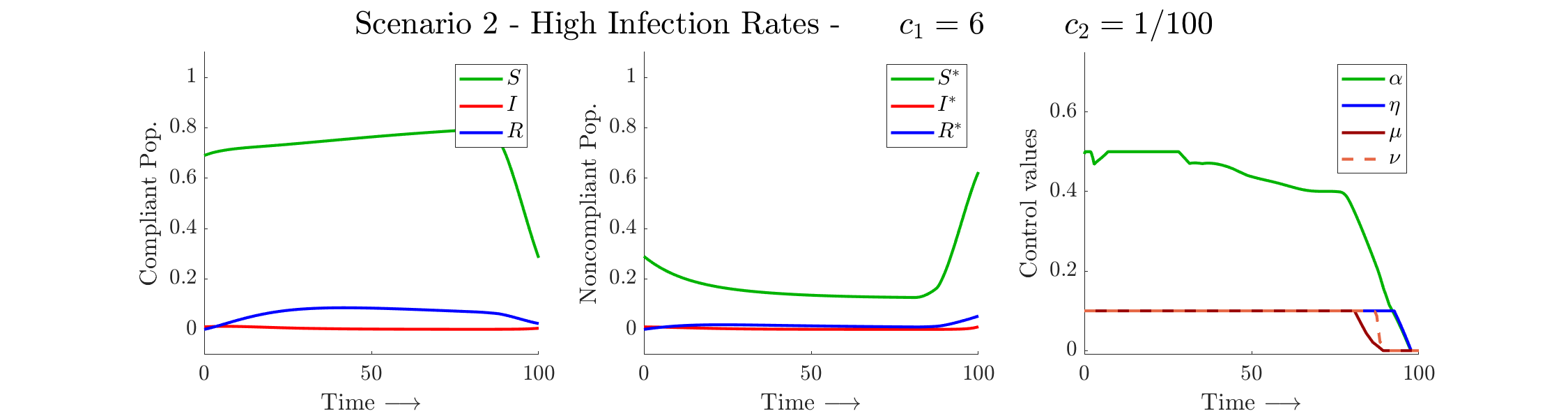}\\
\includegraphics[width=0.85\linewidth]{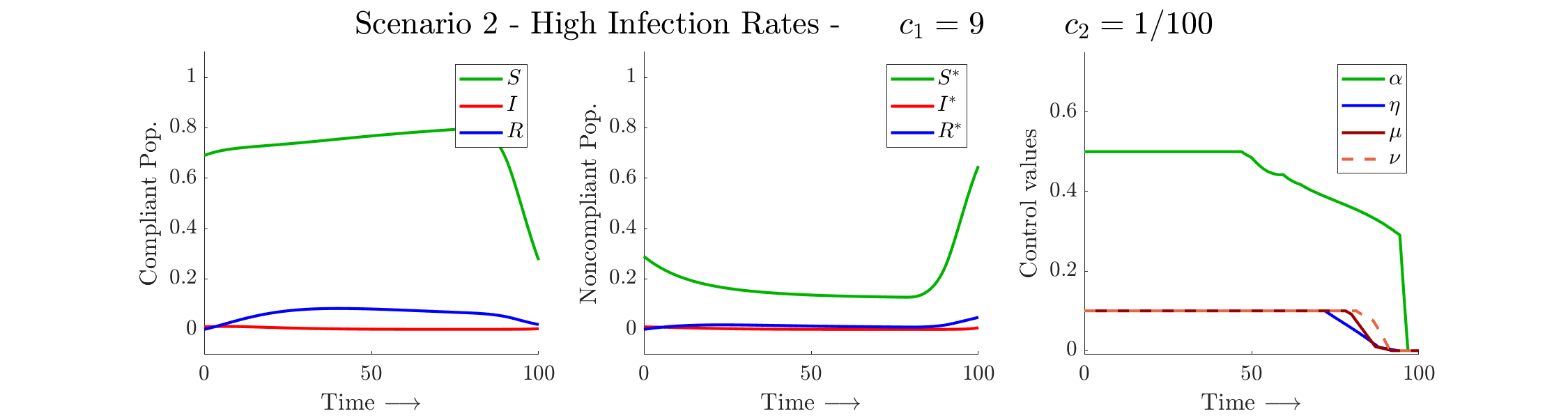}
    \caption{Simulation of Scenario 2, with increased spread of the disease and noncompliance. Noncompliance cost weight is $c_2 = 0.01$. From top to bottom, the cost weight on infections is $c_1 = 1/3, 1, 3, 6, 9$. Because disease spread is so fast, suppression efforts are costly and it is not optimal to fully suppress the disease until $c_1 \ge 6$. In a case like this, an economically oriented policy-maker would allow the disease to spread. } 
    \label{fig:Scenario2}
\end{figure}

\subsection{Scenario 3} Finally, in this scenario, we revert to the parameters of the baseline simulation in figure \ref{fig:Scenario1_balanced}, except that we increase the spread rate for noncompliance to $\mu^* = 0.3$ and observe the effects of changing $\xi$, the proportion of the newly introduced population that is noncompliant. In particular, we have the ``balanced" cost weights $c_1 = 1, c_2 = 0.05$. The full parameter list for scenario 3 is in figure \ref{fig:params}. In some sense, $\xi$ could be interpreted as the baseline proclivity toward noncompliance within the population. A key question is then the level at which this proclivity is high enough that it is no longer optimal to attempt to suppress noncompliance. We include four plots in figure \ref{fig:Scenario3} where $\xi$ takes successive values $0,1/4, 1/2, 1.$ When $\xi = 0$, even with the very high rate of noncompliance spread, it is optimal to significantly quell noncompliance in the early stages of the epidemic. This is seen, for example, in the top panels of figure \ref{fig:Scenario3}, where $\mu$ still takes its maximum value for at the beginning, and noncompliance stays relatively low throughout most of the time interval (we note again that it only grows at the end because the policy-maker is aware of the horizon time, and can save some cost by releasing some controls when there is no longer time for an outbreak to occur). This behavior persists for $\xi = 1/4$, but is contrasted with the ensuing plots when $\xi = 1/2,1$. In those cases, $\mu$ is almost identically zero. The results in the cases $\xi = 1/2$ and $\xi = 1$ are very similar, seeming to indicate that for values $\xi \ge 1/2$, controls are simply ineffective. We note also that $\alpha$ and $\nu$ decrease with increasing $\xi$, which is a result of the fact that disease control is less effective when control of noncompliance cannot be achieved. 

\begin{figure}[t!]
\centering
\includegraphics[width=0.85\linewidth]{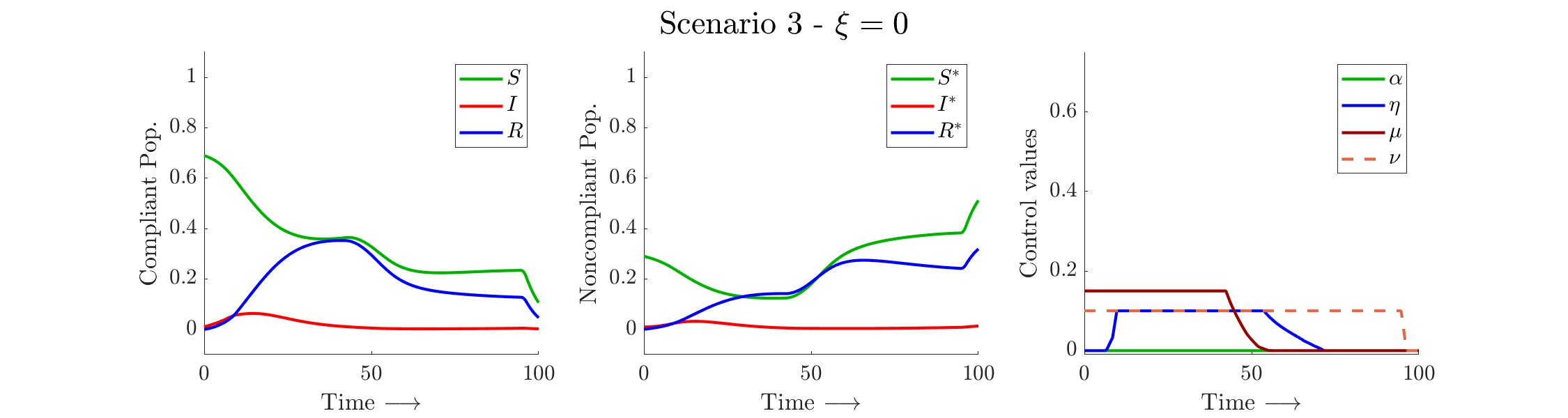}\\
\includegraphics[width=0.85\linewidth]{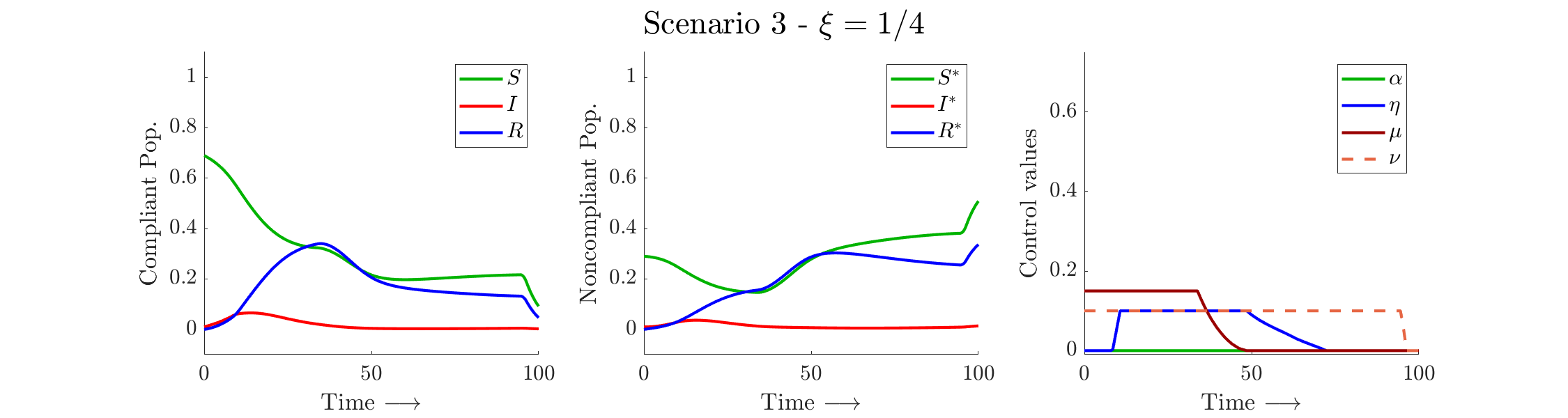}\\
\includegraphics[width=0.85\linewidth]{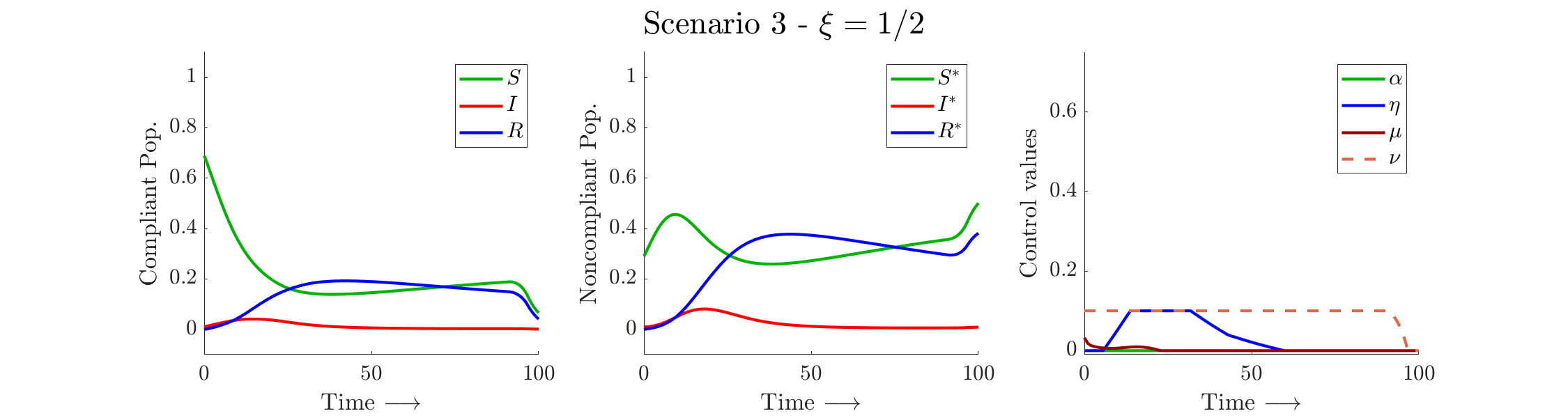}\\
\includegraphics[width=0.85\linewidth]{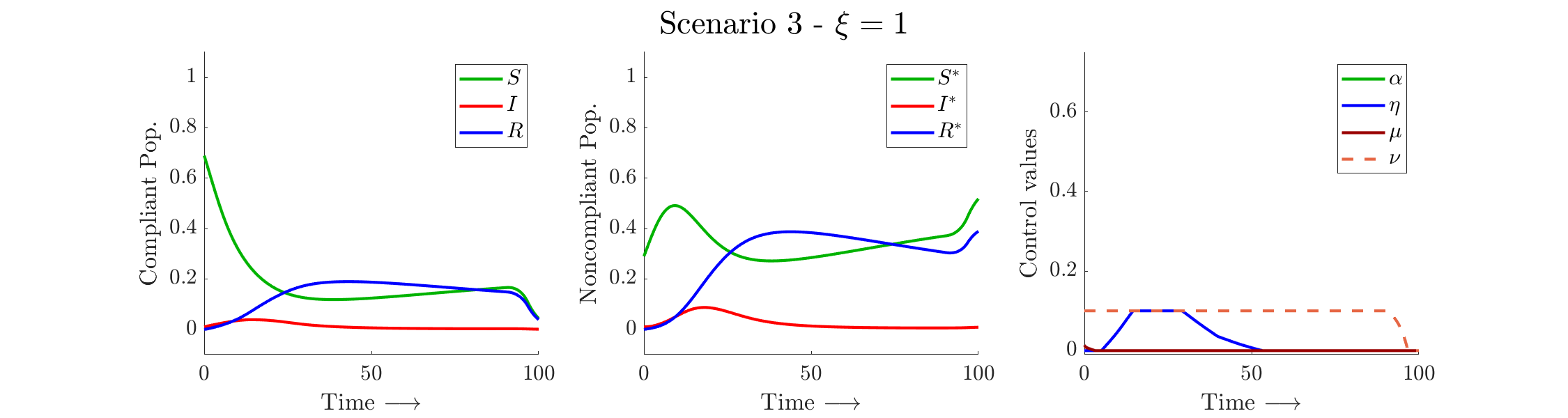}
    \caption{Simulation of Scenario 3, with increased spread rate of noncompliance, balanced cost weights, and variable proportion of the newly introduced population which is noncompliant $(\xi)$. Even with the increased spread of noncompliance, as long as $\xi = 0, 1/4$, the optimal strategy still employs the maximum available $\mu$ for some time to suppress spread of noncompliance. This becomes futile when $\xi \ge 1/2$, so that $\mu$ is essentially zero. There is very little change between the simulations using $\xi=1/2$ and $\xi=1$, indicating that by the time $\xi \ge 1/2$, a threshold has been crossed and controls are ineffective.} 
    \label{fig:Scenario3}
\end{figure}

\section{Conclusion \& Future Directions} \label{sec:conclusion}

In this paper, we present and analyze a compartmental SIR-style model for epidemiology, incorporating noncompliance with governmental NPIs as a social contagion. We provide stability analysis for the disease-free equilibria of the model, and then introduce four control variables that a policy-maker may have access to: (1) reduction in infectivity corresponding to the strength of the NPIs the policy-maker enacts, (2) increase in recovery rate for compliant individuals corresponding to increased treatment efforts, (3) reduction in spread of the social contagion of noncompliance corresponding to public health information campaigns, and (4) ``recovery" from noncompliance corresponding to educational campaigns aimed at increasing compliance. \WW{We prove existence of optimal controls for a reasonably general control structure, which allows for dynamics which are nonlinear in the control variables and $L^1$ control costs. We present and analyze the sequential quadratic Hamiltonian method for numerically resolving optimal controls and argue that this is preferable to other commonly used methods.} Finally, we demonstrate the behavior of our model, focusing specifically on the behavior of the optimal control maps in several parameter regimes that represent different physical scenarios and policy-maker preferences.

\WW{While this study is qualitative, it suggests several sociological takeaways. First, optimal policy depends heavily on both the epidemiological dynamics and the behavioral dynamics. Second, if we believe that human behavior enters into epidemiological models dynamically, then policies aimed toward large scale behavioral modification can be just as important as those aimed at slowing disease spread, and a synergistic approach to control will be more effective. Third, optimal policies are heavily influenced by the preferences of the policy-maker; policy-makers who weigh implementation costs versus public health outcomes differently will employ vastly different control strategies even under identical epidemiological conditions. And fourth, policies aimed at increasing adherence with NPIs have a positive (if indirect) effect on epidemic outcome, and ignoring these could lead to overly simplistic intervention strategies. In fact, interventions targeting noncompliance are justified solely through their downstream effects, even when the policy-maker assigns no direct cost to increased noncompliance. }

We suggest several avenues for future work. A first would be to incorporate stochastic effects to account for various uncertainties inherent to both epidemic spread of diseases and human behavior. This is partially addressed in \cite{PW3}, but not in the context of optimal control, and quantification of the effect that stochasticity has on the optimal control maps could be interesting. Another direction would be to synthesize these results with other modeling frameworks. For example, a multiplex network describing the coevolution of disease and opinion spread is presented in \cite{Peng}, but does not include optimal control of disease spread. Modeling controls in a similar manner to ours but in the network setting could help illuminate the ways in which heterogeneity of contact networks affects spread of the disease in the presence of noncompliant behavior. Next, inclusion of more finely modeled control variables could push models like this toward real-world utility. This could include things like the singular and/or one-shot vaccinations as in \cite{OneShot,MR3012899} respectively, limited quarantine availability in \cite{OC_limitQuarantine}, or control of information via media outlets like in \cite{OC_withMediaCoverage}. Incorporation of finely tuned controllers into models of epidemiology with human behavioral effects could lead to vastly increased modeling fidelity. \WW{Lastly, rather than modeling the spread of noncompliance via mass action, one could use the language of evolutionary game theory, wherein compliance and noncompliance are competing behavioral strategies which spread according to some measure of payoff-driven fitness. Careful considerations of the exact manners in which behaviors spread could lead to models which more faithfully represent the true dynamics. Likewise, more general behavioral modeling wherein compliance-noncompliance is a spectrum rather than a binary state may be more realistic, and the development of models with this philosophy is a promising path forward.} 

\section*{Acknowledgments}





Chloe Ngo was partially supported by the Honors Research Assistant Program (HRAP) of the University of Oklahoma. Weinan Wang was partially supported by the Simons Foundation grant. Both Chloe Ngo and Weinan Wang were partially supported by the Junior Faculty Fellowship from the University of Oklahoma.

\bigskip\noindent 
Chloe Ngo, Honors College, University of Oklahoma, Norman, OK, USA;
 e-mail: \url{chloe.ngo@ou.edu}

\bigskip\noindent 
Christian Parkinson, Department of Mathematics \& Department of Computational Mathematics, Science and Engineering, Michigan State University, East Lansing, MI,  USA;
e-mail: \url{chparkin@msu.edu}

\medskip\noindent 
Weinan Wang, Department of Mathematics, University of Oklahoma, Norman, OK, USA;
 e-mail: \url{ww@ou.edu}

\end{document}